\documentclass[runningheads]{llncs}

\usepackage{mathtools}
\usepackage{amsmath}
\usepackage{hyperref}
\usepackage{xcolor}
\usepackage{amssymb}
\usepackage{stmaryrd}
\usepackage{bm}
\usepackage{dsfont}
\usepackage{graphicx}
\usepackage{tikz}
\usepackage{booktabs}
\usepackage{algorithm}
\usepackage{multirow}
\usepackage{rotating}
\usepackage{array}
\usepackage{algpseudocode}
\usepackage[T1]{fontenc}
\usepackage{pgfplots}
\pgfplotsset{compat=1.18}
\usetikzlibrary{arrows.meta}
\usetikzlibrary{shapes.geometric}
\usetikzlibrary{automata, positioning}

\newcommand{\indi}[1]{\mathds{1}\{#1\}}

\newcommand{\einit}{\hat{\lambda}}
\newcommand{\tinit}{\lambda^{\star}}

\newcommand{\prob}{\mathbb{P}}
\newcommand{\expe}{\mathbb{E}}

\newcommand{\NN}{\mathbb{N}}
\newcommand{\pNN}{\mathbb{N}^+}
\newcommand{\RN}{\mathbb{R}}

\newcommand{\escore}{E}

\newcommand{\lb}{L}
\newcommand{\ub}{U}

\newcommand{\clb}{l}
\newcommand{\cub}{u}

\newcommand{\prop}{\varphi}

\newcommand{\sdistr}{\Delta}

\newcommand{\dist}{d}

\newcommand{\bfoo}{\beta}
\newcommand{\afoo}{\alpha}

\newcommand{\score}{\ell}

\newcommand{\conf}{\delta}
\newcommand{\error}{\varepsilon}

\newcommand{\env}{\theta}

\newcommand{\pol}{\pi}

\newcommand{\evar}{N}

\newcommand{\tenv}{\env^{\star}}
\newcommand{\trm}{P}
\newcommand{\etrm}{\hat{P}}
\newcommand{\ttrm}{P^{\star}}
\newcommand{\tmdp}{\mdp^{\star}}

\newcommand{\inter}{I}

\newcommand{\rQ}{Q}

\newcommand{\rf}{\mathrm{ref}}

\newcommand{\aS}{\mathcal{S}}

\newcommand{\rS}{S}
\newcommand{\cS}{s}

\newcommand{\aA}{\mathcal{A}}

\newcommand{\cA}{a}

\newcommand{\aO}{\mathcal{O}}

\newcommand{\rV}{V}
\newcommand{\cV}{v}

\newcommand{\rW}{W}
\newcommand{\cW}{w}

\newcommand{\aX}{\mathcal{X}}

\newcommand{\rX}{X}
\newcommand{\cX}{x}

\newcommand{\cY}{y}

\newcommand{\reY}{\hat{Y}}
\newcommand{\ceY}{\hat{y}}

\newcommand{\rtY}{Y^{\star}}
\newcommand{\ctY}{y^{\star}}

\newcommand{\aZ}{\mathcal{Z}}

\newcommand{\rZ}{Z}
\newcommand{\cZ}{z}

\newcommand{\emdp}{\hat{\mdp}}
\newcommand{\eenv}{\hat{\env}}

\newcommand{\wscore}{F}

\newcommand{\loss}{\ell}

\newcommand{\moni}{\mathbf{M}}

\newcommand{\sgn}{\sigma}

\newcommand{\mdp}{\mathcal{M}}

\begin{document}
\title{Alignment Monitoring}
%
%
\author{
Thomas A. Henzinger \inst{2}\orcidID{0000-0002-2985-7724} \and 
Konstantin Kueffner\inst{2}\thanks{Part of the work was conducted during an internship at NVIDIA.}\orcidID{0000-0001-8974-2542} \and
Vasu Singh \inst{3}\orcidID{0009-0003-5468-7899}\and
I Sun\inst{1,3}\orcidID{0009-0009-2194-9526} }
\authorrunning{T. Henzinger et al.}
%
\institute{University of Illinois Urbana-Champaign, Illinois, USA, \email{is16@illinois.edu} \and
ISTA, Klosterneuburg, Austria, \email{\{tah, kkeuffner\}@ista.ac.at} \and 
NVIDIA, Santa Clara, USA
\email{{vasus@nvidia.com}}}

\maketitle              
%

\begin{abstract}
Formal verification provides assurances that a probabilistic system
satisfies its specification---\emph{conditioned} on the system model being
aligned with reality.
We propose \emph{alignment monitoring} to watch that this assumption is
justified.
We consider a probabilistic model well aligned if it accurately predicts
the behaviour of an uncertain system in advance.
An \emph{alignment score} measures this by quantifying the similarity between
the model's predicted and the system's (unknown) actual distributions. 
An alignment monitor observes the system at runtime;
at each point in time it uses the current state and the model to predict the
next state.
After the next state is observed, the monitor updates the verdict,
which is a high-probability interval estimate for the true alignment score.
We utilize tools from sequential forecasting to construct our alignment
monitors.
Besides a monitor for measuring the expected alignment score,
we introduce a differential alignment monitor,
designed for comparing two models,
and a weighted alignment monitor,
which permits task-specific alignment monitoring.
We evaluate our monitors experimentally on the PRISM benchmark suite.
They are fast, memory-efficient, and detect misalignment early.
\keywords{Model alignment  \and Runtime verification \and Statistical monitoring.}
\end{abstract}

\section{Introduction}

Probabilistic models such as Markov chains or Markov decision processes (MDPs) can be used to capture uncertainties in a system, for example, about the distribution of inputs.
While many of these models are amenable to formal verification, the verification results are only as good as the models themselves:
if some parameters of the actual system are different from the assumptions made by the probabilistic model, or change over time, then the properties of the model become detached from reality.
Hence, we must ensure that a probabilistic model is an accurate abstraction of the deployed system.
This assurance can only be given by a runtime monitor.

We take the position that all models are ``wrong,''~\cite{box1976science} irrespective of whether we obtain our model from an expert or by learning it from data.
But some models are useful, and we formalize the ``usefulness'' of a model by its power to accurately predict how the system evolves. 
In this paper, we test the assumption that \emph{a given probabilistic model accurately predicts the next state of a system} by constructing quantitative runtime monitors.
In our setting, the system's states are fully observable, but its probabilistic next-state function is unknown. 

\paragraph{Alignment monitoring.}
An alignment monitor outputs, at every point in time, an ``alignment score'' that asserts how well-aligned the model is to the system, by measuring the prediction quality of the model. 
For stochastic models, the prediction quality for a given model and system state translates into some similarity measure between the model's and the system's successor distributions.
Hence, the model's alignment score w.r.t.\ the current history must be some function mapping a sequence of pairs of distributions to a real value.
However, without access to the system's distributions, this history is hidden and must be estimated from observations.
Fortunately, we get these samples for free during runtime.
At every point in time, we interpret the model's current successor distribution as a predictor for the next system state, which will be observed once the system transitions from the current to the next state. 
In alignment monitoring this sequence of predictions and observed states is the observed history and the input to the monitor. 
The observed states are realisations of the system's successor distributions, which are unknown but required for computing the alignment score.
The monitor computes an interval estimate from the {\em observed history} to estimate the model's alignment score evaluated on the {\em hidden history}.
We require that the sequence of interval estimates is correct at all times with high probability.

\paragraph{Average alignment monitor.}
In sequential forecasting scoring rules are used to quantify the quality of a prediction compared to an observation.
The observed outcome is a sample from an unknown true distribution.
The expected score w.r.t.\ to the true distribution provides a measure of how well the prediction is aligned with the true distribution.
This is generalized to the sequential setting by taking the average over all predictions.
This defines the {\em average expected score} (AES).
The AES can be estimated from a sequence of predictions and observations using tools from sequential forecasting, where this problem is well-known~\cite{choe2021comparing}. 
The only assumption required is that the scoring rule be bounded and that the true distribution be determined before the prediction is made. 
We utilize these tools to construct a monitor that outputs, at every point in time, the average expected alignment score. 

\paragraph{Extensions.}
We present two extensions to the average alignment monitor.
The \emph{weighted alignment monitor} allows the user to weigh predictions and outcomes depending on their importance for a particular task.
For example, in probabilistic systems bottom strongly connected components (BSCCs), i.e., sets of states from which the system can never escape, are important, because once the system is in a safe BSCC, we are guaranteed to be safe indefinitely~\cite{baier2008principles}. Therefore, if the monitor detects a transition away from a state that, according to the model, should be inside a BSCC, the safety guarantee of the system is in jeopardy. Hence, such transition should be of greater importance when assessing alignment. We can encode this using the weights of our weighted alignment monitor.
Second, the \emph{differential alignment monitor} allows the user to incorporate a reference model when evaluating the tested model's alignment.
The monitor observes both the tested model's and the reference model's predictions, and computes interval estimates for the average expected score difference between the two models.
If the reference model is already trusted, but may be overly conservative, a differential alignment monitor can be used to decide whether a tested model is better or worse on than the trusted reference model on the observed history.

\paragraph{Experiments.}
We evaluate our average and differential alignment monitor on slightly modified versions of the discrete-time Markov chains from the PRISM benchmark suit~\cite{kwiatkowska2012prism}. 
The computation time of our average alignment monitor depends linearly on the input dimension, i.e., the support of the prediction, but is constant w.r.t. the history. Even if the input dimension is large, e.g., $10^6$, our monitor is fast, requiring roughly $260\mu s$ per iteration. We show that the differential alignment monitor can often decide which model is better after a few 100 observations. We evaluate our weighted alignment monitor on two toy examples, demonstrating its applicability to problems in formal verification.

\paragraph{Contributions.}
\begin{itemize}
\item We introduce and formalise the problem of alignment monitoring.
\item We use tools from Choe et al.~\cite{choe2021comparing} to develop a runtime monitor that tracks the average expected alignment score using high-probability interval estimates.
\item We present weighted and differential alignment monitoring as extension.
\end{itemize}

\section{Preliminaries}
Let $\NN$ be the set of natural numbers, $\pNN$ be the set of natural numbers excluding zero, $\RN$ be the set of real numbers. Let $a,b\in \NN$ such that $a<b$.
We define $[a;b]\coloneqq \{a, a+1,\dots b\}$ as the interval from $a$ to $b$ over the natural numbers and as a shorthand we will use $[b]\coloneqq [1;b]$. Let $\mathit{R}\subseteq\RN$ be a subset of the real-numbers, we denote the set of all interval over $\mathit{R}$ as $\inter(\mathit{R})$.
Given a set countable $\aZ$ and $n\in \NN$ we denote the set of all sequences of length $n$ as $\aZ^n$. The set $\aZ^*\coloneqq \bigcup_{n\in \NN} \aZ^n$ denotes the set of all finite sequences. We define $\sdistr(\aZ)$ as the set of all probability distributions over $\aZ$.

\section{Alignment Monitoring}
\label{sec:problem}
In this section we present the alignment monitoring problem.

\begin{example}
\label{ex:mdp0}
    In quantitative model checking we are given a model of the environment represented as a Markov decision process (MDP) $\emdp\coloneqq (\aS, \aA, \etrm,\einit )$ consisting of a state space $\aS$, a set of actions 
$\aA$, a transition probability function $\etrm\colon \aS\times \aA \to \sdistr(\aS)$ where $\sdistr(\mathcal{Z})$ denotes the set of all probability distributions over a set $\mathcal{Z}$, and an initial distribution $\einit\in \sdistr(\aS) $. 
The nondeterminism introduced by the actions is resolved by a policy $\pol\colon \aS \times \aA$, resulting in a Markov chain (MC). Given a specification $\prop$ represented as an LTL formula, the probability that the policy satisfies the property w.r.t.\ the model is computed, i.e., $\prob_{\pol}^{\emdp}(\rS \models \prop)$. This assures that the system, the environment and the policy, adhere to the specification. The setting above is fundamental to much of the work done in formal verification~\cite{baier2008principles}. It relies on the assumption that model $\emdp\coloneqq (\aS, \aA, \etrm,\einit )$ is aligned with the \emph{actual} environment $\tmdp\coloneqq (\aS, \aA, \ttrm,\tinit )$. 
\end{example}
 We choose to measure the model's alignment with reality based on its predictive capabilities, e.g., how well can it predict the next state.
\begin{example}
    \label{ex:coin}
    Let $p^{\star}$ be the bias of an unknown coin and let $\hat{p}$ be the bias of the model coin given to us. 
    We would consider the model coin well aligned with the actual coin, if $\hat{p}$ and $p^{\star}$ are close, e.g., if $| \hat{p}-p^{\star}|$ is small.
\end{example}
\begin{example}[Ex.~\ref{ex:mdp0} cont.]
    \label{ex:mdp}
    We consider the model $\emdp$ well aligned with reality $\tmdp$ for a state-action pair $(\cS, \cA) \in (\aS\times \aA)$, if the respective successor distributions are ``close'', i.e., if $\etrm(\cS, \cA) \approx \ttrm(\cS, \cA)$ for some notion of similarity. 
\end{example}

\paragraph{Setting.}
Our objective is to monitor the alignment of the model with the environment. 
In monitoring, we are limited to watching reality unfold one step at a time, e.g., the monitor observes the state-action pairs generated by an agent interacting with the environment. We can model this as a stochastic process $\rX\coloneqq (\rX_t)_{t\in \pNN}$ over a given state space $\aX$, e.g., the joint state-action space $\aS\times \aA$. 
The stochastic process is defined by the environment $\tenv\colon \aX^* \to \sdistr(\aX)$ which is modelled as $\eenv\colon \aX^* \to \sdistr(\aX)$.
Together, the environment and the model define two stochastic processes, the hidden process $\rV\coloneqq (\rV_t)_{t\in \pNN} = (\reY_t, \rtY_t)_{t\in \pNN}$ and the observed process $\rW\coloneqq (\rW_t)_{t\in \pNN} = (\reY_t, \rX_t)_{t\in \pNN}$ where for every $t\in \pNN$ 
\begin{align*}
  \reY_t  = \eenv(\rX_1, \dots, \rX_{t-1}) \quad \text{and} \quad \rX_t \sim \rtY_t = \tenv(\rX_1, \dots, \rX_{t-1}).
\end{align*}
The hidden history $\cV\coloneqq \cV_1, \dots, \cV_t$ at time $t$ is a finite realisation of $\rV_1, \dots, \rV_t$ consisting of model's predictions and environment's probability distribution. 
The observed history $\cW\coloneqq \cW_1, \dots, \cW_t$ at time $t$ is a finite realisation of $\rW_1, \dots, \rW_t$ consisting of model's predictions and observed states.
We can summarise the dynamics as follows. At time $t$ the environment decides on the distribution $\ctY_t=\tenv(\cX_1, \dots, \cX_{t-1})$, a prediction is made using the model $\ceY_t=\eenv(\cX_1, \dots, \cX_{t-1})$, after which the next state $\cX_{t}$ is obtained by sampling from $\ctY_t$.


\begin{example}[Ex.~\ref{ex:coin} cont.]
    \label{ex:coin2}
    For the coins $p^{\star}$ and $\hat{p}$, the observed process $(\reY_t, \rX_t)_{t\in \pNN}$ and the hidden process $(\reY_t, \rtY_t)_{t\in \pNN}$ are defined for every $t\in \pNN$ such that $\rtY_t=\mathrm{Bernoulli}(p^{\star})$ and $\reY_t=\mathrm{Bernoulli}(\hat{p})$ a.s., and $\rX_t\sim \rtY_t$.
\end{example}

\begin{example}[Ex.~\ref{ex:mdp} cont.]
    \label{ex:mdp2}
    Because $\pol$ is deterministic we focus only on the states, i.e., $\aX= \aS$. The observed process $(\reY_t, \rX_t)_{t\in \pNN}$ and the hidden process $(\reY_t, \rtY_t)_{t\in \pNN}$ are defined for every $t\in \pNN$ such that $\reY_{t+1}= \etrm(\rS_t, \pol(\rS_t))$, $\rtY_{t+1}= \ttrm(\rS_t, \pol(\rS_t))$, and $\rX_{t}= \rS_{t+1} \sim \rtY_{t+1}$ with $\reY_1 = \einit$ and $\rtY_1= \tinit$.
\end{example}

\paragraph{Alignment monitoring.}
Without any assumption on the environment $\tenv$, little can be said about the overall alignment of $\eenv$, by looking only at the observed history. We can, however, say something about the alignment thus far. 
Here we follow the intuition: if the model has shown a consistent track record of predicting the future, there is no reason to reject it;
if the model consistently failed to predict the future, we should reject it. Formally, we measure the alignment of the model at every point in time using an alignment score function $\dist\colon (\sdistr(\aX) \times \sdistr(\aX))^*\to \RN$, which computes an alignment score between two sequences of distributions. Our objective is to construct a monitor $\moni\colon (\sdistr(\aX)\times \aX)^* \to \inter(\RN)$ computing an interval $[\clb_t, \cub_t]=\moni(\cW_1,\dots, \cW_t)$ from the observed history $\cW_1,\dots, \cW_t$ at time $t\in \pNN$ estimating the quantity $\dist(\cV_1,\dots,\cV_t)$ evaluated over the hidden history $\cV_1,\dots,\cV_t$ with high probability.
\begin{problem}
    \label{prob:fore}
    Given an unknown environment $\tenv\colon \aX^*\to \sdistr(\aX)$, a model $\eenv\colon \aX^*\to \sdistr(\aX)$, an alignment score function $\dist\colon (\sdistr(\aX) \times \sdistr(\aX))^*\to \RN$, and an error probability $\conf\in (0,1)$, find a monitor $\moni\colon (\sdistr(\aX)\times \aX)^* \to \inter(\RN)$ such that 
    \begin{align}
        \label{eq:sound}
        \prob_{\tenv}\left( \forall t\in \pNN\colon \dist(\rV_1, \dots, \rV_t) \in \moni(\rW_1,\dots,  \rW_t)\right) \geq 1-\conf.
    \end{align}
\end{problem}
The condition $\forall t\in \pNN\colon \dist(\rV_1, \dots, \rV_t) \in \moni(\rW_1, \dots, \rW_t)$ is an invariant guaranteeing that the monitor bounds the alignment score at every iteration during its infinite run. Equation~\ref{eq:sound} requires this invariant to hold with high-probability.


\section{Average Alignment Monitor}
\label{sec:alignment}
In this section we present a monitor solving Problem~\ref{prob:fore} for the average alignment score, i.e., the average over individual alignment scores computed using scoring rules. The monitor is based on tools from the sequential forecasting literature~\cite{choe2021comparing}.

\paragraph{Scoring rule.}
A bounded scoring rule is a function $\score\colon \sdistr(\aX) \times \aX \to [a,b]$ for $a,b\in \RN$ assessing the quality of the model's prediction at every time step w.r.t.\ the observation. The example below highlights two common scoring rules.
\begin{example}
    \label{ex:score_foo}
    For a distribution $\cY\in \sdistr(\aX)$, and an observation $\cX\in \aX$, the Brier score $\score_{B}$, bounded on $[0,2]$, and the spherical score $\score_{S}$, bounded on $[-1,0]$, are 
\begin{align*}
   \score_{B}(\cY, \cX) \coloneqq  \sum_{\cX'\in \aX} (\cY(\cX') - \indi{\cX'=\cX})^2 \quad \text{and} \quad  \score_{S}(\cY, \cX)\coloneqq \frac{-\cY(\cX)}{ \sqrt{\sum_{\cX' \in \aX}\cY(\cX')^2 }}.
\end{align*}
\end{example}
To assess the prediction quality w.r.t.\ the environment, we need to compute the expected score $\expe_{\rX_t\sim \ctY_t}(\score(\ceY_t, \rX_t))$. To ensure fair scoring, it is important that the scoring rule is proper, i.e., the expected score is minimised when the predicted distribution matches the true distribution, i.e., if all $\ctY, \ceY \in \sdistr(\aX)$,
\begin{align*}
    \expe_{\rX \sim \ctY}[\score(\ctY, \rX)] \leq \expe_{\rX \sim \ctY}[\score(\ceY, \rX)].
\end{align*}
The expected score assesses the alignment of the model at every $t\in \pNN$.
The average expected score (AES) extends this to sequences defined; it is defined over the hidden history $\cV\coloneqq \ceY_{1}, \ctY_1, \dots, \ceY_t, \ctY_t$ as
\begin{align*}
    \escore^{\score}(\cV) \coloneqq \frac{1}{t}\sum_{i=1}^t \expe_{\rX_i\sim \ctY_i}(\score(\ceY_i, \rX_i)).
\end{align*}

\subsection{Monitor Construction}
Computing the AES requires knowledge of the successor distributions as given by the environment. Because this is hidden from us, our monitor must estimate the AES during runtime. 


\paragraph{Point estimation.}
A natural estimator for the AES is the average score defined for every observed history $\cW\coloneqq \ceY_{1}, \cX_1, \dots, \ceY_t, \cX_t$ as
\begin{align*}
    \hat{\escore}^{\score} (\cW) \coloneqq \frac{1}{t}\sum_{i=1}^t \score(\ceY_i, \cX_i).
\end{align*}

\paragraph{Confidence sequences.}
We use confidence sequences to quantify how close our average score $\hat{\escore}_t\coloneqq \hat{\escore}^{\score}(\rW_1, \dots, \rW_t) $ is to the AES $\escore_t\coloneqq \escore^{\score}(\rV_1, \dots, \rV_t) $ for every $t\in \pNN$.
A confidence sequence for the sequence of AESs $(\escore_t )_{t\in \pNN}$, is a sequence of lower and upper bounds $(\lb_t , \ub_t )_{t\in \pNN}$ containing $(\escore_t )_{t\in \pNN}$ with high probability, i.e., for $\conf\in (0,1)$ the confidence sequence $(\lb_t , \ub_t )_{t\in \pNN}$ ensures
\begin{align*}
    \prob(\forall t\in \pNN\colon \escore_t \in [\lb_t, \ub_t] ) \geq 1-\conf.
\end{align*}
Using techniques outlined in Howard et al.~\cite{howard2021time} and Choe et al.~\cite{choe2021comparing} we construct a confidence sequence centred around the point estimate, i.e., at every time $t\in \pNN$ the lower and upper bounds are defined as $\lb_t\coloneqq \hat{ \escore}_t-\error_t$ and $\ub_t\coloneqq \hat{\escore}_t+\error_t$ respectively. The error $\error_t$ is given by 
\begin{align}
    \label{eq:bound}
    \error_t(\evar_t,\conf, \sgn_{\score}) \coloneqq \frac{\sqrt{2.13 \cdot \evar_t \cdot g(\evar_t,\conf) + 1.76\cdot \sgn_{\score}^2 \cdot g(\evar_t,\conf)^2}+ 1.33 \cdot \sgn_{\score} \cdot g(\evar_t,\conf)}{t} 
\end{align}
where $g(n,\conf)=\left(2 \cdot \log\left(\pi\log(n) / \sqrt{6}\right) + \log(2/\conf)\right) $, $\sgn_{\score}\coloneqq b-a$ is the difference between the maximal and minimal value of the scoring rule, and $\evar_t$ is the maximum between $1$ and the empirical variance process, i.e., 
\begin{align}
    \evar_t \coloneqq \max\left(1,  \sum_{i=1}^t (\score(\reY_i,\rX_i ) - \hat{\escore}_{i-1})^2\right).
\end{align}

\paragraph{Implementation.}
The alignment monitor implemented in Algorithm~\ref{alg:align} requires constant space and time w.r.t. the observed history. 
\emph{Space:} the monitor uses three counters to incrementally compute the time $t$, the empirical variance process $\evar$, and the average score $\hat{E}$. 
\emph{Time:} the monitor requires constant time to update the three counters. The only computationally demanding operation is computing the score stored in variable $s$. 
This depends on the scoring rule. For example, in the case of the Brier score and the spherical score this is in the order of $\mathcal{O}(|\aX|)$. 
We denote $T_{\score, \aX}$ as the time required to evaluate the scoring rule. 

\begin{theorem}
    \label{theorem:general_alignment}
    Let $\rV$ and $\rW$ be the hidden and the observed process defined by the environment $\tenv$ and the model $\eenv$. Let $\score$ be a scoring rule bounded on the interval $[a,b]\subset\RN$, and let $\conf\in (0,1)$ be an error probability threshold, then the monitor $\moni_{\score, \conf}$ solves Problem~\ref{prob:fore} for $\dist\coloneqq \escore^{\score}$.
The monitor requires at each iteration $\aO(1)$-space and $\aO(T_{\score, \aX})$-time w.r.t.\ the history and the state space, where $T_{\score, \aX}$ is the time required to evaluate the scoring rule. 
\end{theorem}

\begin{algorithm}
\caption{Average Alignment Monitor $\moni_{\score, \conf}$}
\label{alg:align}
\begin{algorithmic}[1]
\Require Error probability $\conf \in (0,1)$, scoring rule $\score\colon \sdistr(\aX)\times \aX \to [a,b]$.
\Function{Init}{}
    \State $t \gets 0$;  \;$\hat{\escore} \gets 0$;  \; $\evar \gets 1; \;  \sgn_{\score}\gets b-a $
\EndFunction

\Function{Next}{$\ceY, \cX$}
    \State $s \gets \score(\ceY, \cX)$;  \;$t \gets t + 1$;  \; $\evar \gets \max(1, \evar + (s - \hat{\escore})^2)$
    \State $\hat{\escore} \gets \frac{1}{t}\cdot (t - 1) \cdot \hat{\escore} + s$; \;  $g \gets 2 \log\left( \frac{\pi \log(\evar)}{\sqrt{6}} \right) + \log\left( \frac{2}{\conf} \right)$
    \State $\error \gets \frac{ 1 }{t}\left( \sqrt{2.13 \cdot \evar \cdot g + 1.76 \cdot \sgn_{\score}^2 \cdot g^2} + 1.33 \cdot \sgn_{\score} \cdot g\right)$
    \State \Return $[\hat{\escore} - \error, \hat{\escore} + \error]$
\EndFunction
\end{algorithmic}
\end{algorithm}

\section{Extensions}
In Section~\ref{sec:alignment} we presented a monitor for the average expected score (AES), i.e., the average of all past scores computed by a scoring function of a single model. 
In this section we extend our average alignment monitor by
the differential alignment monitor, which compares the alignment scores of two models, and 
the weighted alignment monitor, which computes the weighted average of weighted scores.


\subsection{Differential Alignment Monitor}
Monitoring a single quantitative value may not be overly informative without a reference point. This is where differential alignment monitoring comes in.

\paragraph{Differential alignment monitoring.}
Assume that in addition to the model $\eenv$, we are given a reference model $\eenv^{\rf}$. The reference model represents a benchmark against which we want to assess the performance of $\eenv$. 
\begin{example}[Ex~\ref{ex:mdp2} cont.]
\label{ex:mdp3}
    The reference model $\emdp^{\rf}\coloneqq (\aS,\aA, \etrm^{\rf}, \einit^{\rf})$ for the environment $\tmdp$ differs depending on the available knowledge. 
    If little is known, a worst-case reference model is the uniform distribution, i.e., the model should be at least better than random chance.
    We distinguish between a black- and gray-box setting.
    In the black-box setting, $\emdp^{\rf}$ assigns \emph{each state} the same probability, i.e., for all $\cS,\cS'\in \aS$ and $\cA\in \aA$ we have $\etrm^{\rf}(\cS,\cA, \cS')=1/|\aS|$.
    In the gray-box setting, $\emdp^{\rf}$ assigns \emph{each successor} the same probability, i.e., for all $\cS, \cS'\in \aS$ and $\cA\in \aA$ we have $\etrm^{\rf}(\cS,\cA,\cS')=1/|\aS_{\cS,\cA}|$ if $\cS'\in \aS_{\cS, \cA}$ where $\aS_{\cS,\cA}\coloneqq \{\cS'\in \aS\mid \ttrm(\cS, \cA,\cS')>0\}$, else $\etrm^{\rf}(\cS,\cA,\cS')=0$.
\end{example}
In the differential alignment monitoring setting, the monitor observes the predictions of both the model and the reference model at the same time, i.e., the hidden process 
$\rV^{\rf}\coloneqq (\rV_t^{\rf})_{t\in \pNN}= (\reY_t,\reY_t ^{\rf}, \rtY_t)_{t\in \pNN}$ and the observed process $\rW^{\rf}\coloneqq (\rW_t^{\rf})_{t\in \pNN}= (\reY_t,\reY _t^{\rf}, \rX_t)_{t\in \pNN}$ are defined analogously to $\rV$ and $\rW$.
\begin{problem}
    \label{prob:diff_fore}
    Given an unknown environment $\tenv\colon \aX^*\to \sdistr(\aX)$, a model $\eenv\colon \aX^*\to \sdistr(\aX)$, a reference model $\eenv^{\rf}\colon \aX^*\to \sdistr(\aX)$, an alignment score function $\dist\colon (\sdistr(\aX) \times \sdistr(\aX))^*\to \RN$, and an error probability $\conf\in (0,1)$, find a monitor $\moni\colon (\sdistr(\aX)\times \sdistr(\aX)\times \aX)^* \to \inter(\RN)$ such that 
    \begin{align}
        \label{eq:diff_sound}
        \prob_{\tenv}\left( \forall t\in \pNN\colon  D(\rV_1^{\rf}, \dots, \rV_t^{\rf}) \in \moni(\rW_1^{\rf}, \dots ,\rW_t^{\rf}) \right) \geq 1-\conf
    \end{align}
    where $D(\rV_1^{\rf}, \dots, \rV_t^{\rf}) = \dist(\reY_1, \rtY_1, \dots, \reY_t, \rtY_t)- \dist(\reY_1^{\rf}, \rtY_1, \dots, \reY_t^{\rf}, \rtY_t)$.
\end{problem}

\paragraph{Differential alignment monitor.}
We modify Algorithm~\ref{alg:align} to solve Problem~\ref{prob:diff_fore}. 
The modification, presented in Algorithm~\ref{alg:diff_align}, is limited to computing the score difference and adjusting the score value bounds to $[a-b, b-a]$.
\begin{theorem}
    \label{theorem:diff_alignment}
    Let $\rV^{\rf}$ and $\rW^{\rf}$ be the hidden and the observed process defined by the environment $\tenv$, the model $\eenv$, and the reference model $\eenv^{\rf}$. Let $\score$ be a scoring rule bounded on the interval $[a,b]\subset\RN$, and let $\conf\in (0,1)$ be an error probability, then the monitor $\moni_{\score, \conf}^D$ solves Problem~\ref{prob:diff_fore} for $\dist\coloneqq \escore^{\score}$.
 The monitor requires at each iteration $\aO(1)$-space and $\aO(T_{\score, \aX})$-time w.r.t.\ the history and the state space, where $T_{\score, \aX}$ is the time required to evaluate the score function. 
\end{theorem}

\begin{algorithm}
\caption{Differential Alignment Monitor $\moni_{\score, \conf}^D$}
\label{alg:diff_align}
\begin{algorithmic}[1]
\Require Error probability $\conf \in (0,1)$, scoring rule $\score\colon \sdistr(\aX)\times \aX \to [a,b]$. \Comment{$\sgn_{\loss}$ is scaled by $2$.}
\Function{Init}{}
    \State $t \gets 0$;  \;$\hat{\escore} \gets 0$;  \; $\evar \gets 1; \;  \sgn_{\score}\gets 2\cdot (b-a )$
    \Comment{$\sgn_{\loss}$ is scaled by $2$.}
\EndFunction
\Function{Next}{$\ceY, \ceY^{\rf}, \cX$}
    \State $s \gets \score(\ceY, \cX)- \score(\ceY^{\rf}, \cX)$ ;  \;$t \gets t + 1$ \Comment{$s$ is the score difference.}
    \State $\evar \gets \max(1, \evar + (s - \hat{\escore})^2)$ ;\; $\hat{\escore} \gets \frac{1}{t}\cdot (t - 1) \cdot \hat{\escore} + s$
    \State \Return $[\hat{\escore} - \error_t(\evar, \conf, \sgn_{\loss}), \hat{\escore} + \error_t(\evar, \conf, \sgn_{\loss})]$
\EndFunction
\end{algorithmic}
\end{algorithm}

\subsection{Weighted Alignment Monitor}
The AES treats all predictions and all observations equally. In formal verification, this is not necessarily true. Some predictions may be high-stake. Some wrongly predicted outcomes are worse. Our monitor should be able to account for that. 

\begin{example}
\label{ex:fair}
The Markov chain below encodes a classical bank loan example from the fairness literature~\cite{henzinger2023monitoring}. From the initial state ($S$) either a person from group $A$ or group $B$ is applying for a loan. If the loan is granted ($G$), the person can either repay the loan ($R$) or default on it ($D$). In all other cases we return back to the initial state $S$.
For a common fairness measure, such as the difference of the loan grant probabilities between groups~\cite{henzinger2023monitoring}, we can evaluate the fairness from the model directly, i.e., $\prob(G\mid A)- \prob(G\mid B) = 0.7 - 0.4$. Naturally, we care more about the model's alignment on states $A$ and $B$.
\begin{center}
\scalebox{0.7}{\begin{tikzpicture}[->, >=stealth, auto, node distance=2.5cm, scale=0.8, semithick]

  \tikzstyle{every state}=[draw=black, thick, minimum size=6mm]
  \tikzstyle{dashed state}=[state, draw=black, dashed]

  \node[dashed state] (s0) {$S$};
  \node[state] (s1) [right of=s0, yshift=1cm] {$A$};
  \node[state] (s2) [right of=s0, yshift=-1cm] {$B$};
  \node[state] (s3) [right of=s1] {$G$};
  \node[state] (s5) [right of=s3, yshift=0.5cm] {$R$};
  \node[state] (s6) [right of=s3, yshift=-0.5cm] {$D$};
  \node[state] (s7) [right of=s2] {$G$};
  \node[state] (s9) [right of=s7, yshift=0.5cm] {$R$};
  \node[state] (s10) [right of=s7, yshift=-0.5cm] {$D$};
  \node[dashed state] (s11) [right of=s0, xshift=8cm] {$S$};

  \path (s0) edge node[above] {0.8} (s1)
            edge node[below] {0.2} (s2)
        (s1) edge node[above] {0.7} (s3)
             edge[bend right=40] node[above] {0.3} (s0)
        (s3) edge node[above] {0.3} (s5)
             edge node[below] {0.7} (s6)
        (s2) edge node[above] {0.4} (s7)
             edge[bend left=40] node[below] {0.6} (s0)
        (s7) edge node[above] {0.9} (s9)
             edge node[below] {0.1} (s10)
        (s5) edge (s11)
        (s6) edge (s11)
        (s9) edge (s11)
        (s10) edge (s11);

\end{tikzpicture}}
\end{center}

\end{example}
Quantifying the importance of states is a known concept~\cite{pranger2024test,pouget2021ranking}. 
However, as shown in the example below, it is necessary to consider transitions as well. 
\begin{example}
\label{ex:safe}
    Consider the Markov chain depicted below. The solid lines indicate the model; the environment is the union of the dashed and solid lines.
    Hence, during monitoring we will eventually observe two transitions not in the support of the model. Those two transitions are not equal.     
    Take $s_4$ to be an unsafe state. 
    In our model the probability of being safe when starting in $s_1$ is $0.9$, while in reality the probability is $0$. 
    During monitoring we may observe the system transitioning from $s_6$ to $ s_5$. Although this transition is not in the support of our model, the value of our verdict is not jeopardised. By contrast, if we observe the transition from $s_6$ to $s_2$ we should be alarmed, as it is vital for the validity of our verdict. Naturally, we care more about the latter than the former transition.
    \begin{center}
\scalebox{0.7}{\begin{tikzpicture}[->, >=stealth, auto, node distance=2.5cm, scale=0.8, semithick]

  \tikzstyle{every state}=[draw=black, thick, minimum size=6mm]

  \node[state] (q1) {$s_1$};
  \node[state] (q2) [right of=q1] {$s_2$};
  \node[state] (q3) [above right of=q2, xshift=1cm, yshift=-1cm] {$s_3$};
  \node[state] (q4) [below right of=q2, xshift=1cm, yshift=1cm] {$s_4$};
\node[state] (q5) [right of=q3] {$s_5$};
\node[state] (q6) [right of=q5] {$s_6$};

  \path 
        (q1) edge (q2)
        (q3) edge (q5)
        (q5) edge (q6)
        (q2) edge node[above left] {$0.9$} (q3)
        (q2) edge node[below left] {$0.1$} (q4)
        (q4) edge [bend left=20] (q1);
    \draw[bend right=35] (q6) to (q3);
    \draw[dashed, bend right=30] (q6) to (q5);
    \draw[dashed, bend left=10] (q6) to (q2);
\end{tikzpicture}}
\end{center}

\end{example}

\paragraph{Weighted alignment monitoring.}
Assume the alignment score depends not only on the distributions, but also on the past observations, i.e., $\dist\colon (\sdistr(\aX) \times \sdistr(\aX) \times \aX)^*\to \RN$. 
Our objective is to construct a monitor $\moni\colon (\sdistr(\aX)\times \aX)^* \to \inter(\RN)$ computing an interval $[\clb_t, \cub_t]=\moni(\cW_1,\dots, \cW_t)$ from the observed history $\cW_1,\dots, \cW_t$ at time $t\in \pNN$ estimating, with high probability, the quantity $\dist(\cV_1, \cX_1,\dots,\cV_t, \cX_t)$ evaluated over the hidden history $\cV_1, \dots, \cV_t$ and the observed states $\cX_1, \dots, \cX_t$.
\begin{problem}
    \label{prob:weighted_fore}
    Given an unknown environment $\tenv\colon \aX^*\to \sdistr(\aX)$, a model $\eenv\colon \aX^*\to \sdistr(\aX)$, a weighted alignment score function $\dist\colon (\sdistr(\aX) \times \sdistr(\aX) \times \aX)^*\to \RN$, and an error probability $\conf\in (0,1)$, find a monitor $\moni\colon (\sdistr(\aX)\times \aX)^* \to \inter(\RN)$ s.t.\
    \begin{align}
        \label{eq:weight_sound}
        \prob_{\tenv}\left( \forall t\in \pNN\colon \dist(\rV_1, \rX_1,\dots, \rV_t, \rX_t) \in \moni(\rW_1,\dots,  \rW_t)\right) \geq 1-\conf.
    \end{align}
\end{problem}

\paragraph{Weighted scoring rules.}
Analogously to the average alignment monitor, the weighted alignment monitor uses weighted scoring rules, which are scoring rules $\score_{\omega}\colon \sdistr(\aX)\times \aX \to [c_{\omega}\cdot a, c_{\omega} \cdot b]$ defined by a weight function $\omega\colon \aX\to [0, c_{\omega}]$ for $ c_{\omega}>0$.
There are multiple approaches for creating weighted scoring rules~\cite{allen2024weighted,holzmann2016weighted,allen2023evaluating}.
The example below shows the outcome-based method of Holzmann et al.~\cite{holzmann2016weighted}.
\begin{example}
\label{ex:weighted_rules}
    Given a proper scoring rule $\score\colon \sdistr(\aX)\times \aX \to \RN$ and a weight function $\omega\colon \aX \to [0,1]$, we obtain a scoring rule $\score_{\omega}$ proper on $\{\cX \in \aX\mid \omega(\cX)>0\}$ by defining for $\cY\in \sdistr(\aX)$ and $\cX\in \aX$ 
    \begin{align*}
        \score_{\omega}(\cY, \cX) \coloneqq \omega(\cX) \score(\cY_{\omega}, \cX) \quad \text{where}\quad  \cY_{\omega}(\cX')\coloneqq \frac{\omega(\cX')\cY(\cX')}{\sum_{\cX''\in \aX}\omega(\cX'')\cY(\cX'') } \quad \forall \cX''\in \aX.
    \end{align*}
\end{example}

\paragraph{Weighted alignment score.}
The weighted alignment score is the weighted average over scores computed by weighted scoring rules. 
For $c_{\afoo}, c_{\bfoo}>0$ let $\afoo\colon \aX^* \to [0,c_{\afoo}]$ be a function that assigns each prediction a weight based on the history, and let $\bfoo\colon \aX^*\to (\aX \to [0,c_{\bfoo}])$ be a function that defines a weight function for each outcome based on the history. 
Let $\cW\coloneqq \ceY_1, \cX_1, \dots, \ceY_t, \cX_t$ be a observed history, $\cZ\coloneqq \cX_1, \dots, \cX_t$ the corresponding sequence of states, and $\cZ_{1:k} \coloneqq \cX_1, \dots, \cX_{k}$ the prefix of $\cZ$ of length $k\in [t]$. We define the weighted alignment score as
\begin{align*}
\hat{F}_{\afoo,\bfoo}^{\score}(\cW_1, \dots, \cW_t)\coloneqq \frac{1}{ t_{\afoo}(\cZ)} \sum_{i=1}^t \afoo(\cZ_{1:i-1})\cdot\score_{\bfoo(\cZ_{1:i-1})}(\ceY_i, \cX_i)
\end{align*}
where $\score_{\bfoo(\cZ_{1:i-1})}$ is a weighted scoring function and $t_{\afoo}(\cZ)\coloneqq \sum_{i=1}^t \afoo(\cZ_{1:i-1})$ is the weighted ``progression of time''. 
This is an estimator of the weighted expected score. It is defined for the corresponding hidden history $\cV\coloneqq \ceY_1, \ctY_1, \dots, \ceY_t, \ctY_t$ as
\begin{align*}
\wscore_{\afoo,\bfoo}^{\score}(\cV_1,\cX_1, \dots, \cV_t,\cX_t )\coloneqq \frac{1}{ t_{\afoo}(\cZ)} \sum_{i=1}^t \afoo(\cZ_{1:i-1})\cdot \expe_{\rX_i\sim \ctY_i}(\score_{\bfoo(\cZ_{1:i-1})}(\ceY_i, \rX_i)).
\end{align*}

\begin{example}[Ex.~\ref{ex:safe} and ~\ref{ex:fair}]
    \label{ex:weighted}
    We present weight functions for the Markov chain in Example~\ref{ex:safe} and ~\ref{ex:fair}. We define them as a function of the current state and current transition, i.e., $\afoo \colon \aS\to [0,1]$ and $\bfoo\colon \aS\to (\aS\to [0,1])$.
    In Example~\ref{ex:fair}, we limit our alignment monitor to the states $A$ and $B$, i.e., $\afoo(A)=\afoo(B)=1$ and $0$ otherwise.
    In Example~\ref{ex:safe}, we focus on the states in the bottom strongly connected component (BSCC), i.e., for all 
$s\in C \coloneqq \{s_3, s_5, s_6\}$ we have $\afoo(s)=1$ and $0.1$ otherwise. Moreover, we penalise transitions away from the BSCC, i.e., $\bfoo(s)(s')=1$ if $s\in C$ and $s'\not\in C$, and $0.05$ otherwise. 
\end{example}

\paragraph{Weighted alignment monitor.}
We modify Algorithm~\ref{alg:align} to solve Problem~\ref{prob:weighted_fore} for $\dist=\wscore_{\afoo,\bfoo}^{\score}$. We modify the almost sure bound to be $\sgn_{\score}^{\afoo, \bfoo}\coloneqq c_{\afoo} \cdot c_{\bfoo}  \cdot (b-a)$, because multiplying the score changes its scale. We keep track of the weighted time $t_{\afoo}$ instead of the actual time, and have to remember the entire history because of the weight functions.
Instead of the normal score, we compute the score weighted by both $\afoo$ and $\bfoo$. 
Details are in Algorithm~\ref{alg:weight_align}.
\begin{theorem}
    \label{theorem:weighted_alignment}
    Let $\rV$ and $\rW$ be the hidden and the observed process defined by the environment $\tenv$ and the model $\eenv$. Let $\score$ be a scoring rule bounded on the interval $[a,b]\subset\RN$, let $\afoo\colon \aX^* \to [0,c_{\afoo}]$ and $\bfoo\colon \aX^*\to (\aX \to [0,c_{\bfoo}])$ be a prediction and an observation weight function, and $\conf\in (0,1)$ be an error probability, then the monitor $\moni_{\score, \conf}^W$ solves Problem~\ref{prob:weighted_fore} for $\dist\coloneqq \wscore_{\afoo,\bfoo}^{\score}$.
 The monitor requires at each iteration $\aO(t)$-space and $\aO(t + T_{\score, \aX})$-time w.r.t.\ the history $\cX_1, \dots, \cX_t$ and the state space, where $T_{\score, \aX}$ is the time required to evaluate the score function.
\end{theorem}


\begin{algorithm}
\caption{Weighted Alignment Monitor $\moni_{\score, \conf}^W$}
\label{alg:weight_align}
\begin{algorithmic}[1]
\Require Error probability $\conf \in (0,1)$, scoring rule $\score\colon \sdistr(\aX)\times \aX \to [a,b]$, prediction and observation weight function  $\afoo\colon \aX^* \to [0,c_{\afoo}]$ and $\bfoo\colon \aX^*\to (\aX \to [0,c_{\bfoo}])$
\Function{Init}{}
    \State $t \gets 0$;  \;$\hat{\escore} \gets 0$;  \; $\evar \gets 1$;\; $\cZ \gets \epsilon$ \Comment{Initialise Memory}
    \State $\sgn_{\score}^{\afoo, \bfoo}\gets c_{\afoo} \cdot c_{\bfoo}  \cdot (b-a)$
\Comment{$\sgn_{\loss}$ scaled by max weights.}
\EndFunction
\Function{Next}{$\ceY, \cX$}
    \State $s \gets \afoo(\cZ)\cdot  \score_{\bfoo(\cZ)}(\ceY, \cX)$; \Comment{$s$ is the score times the weights.}
    \State $t \gets t + \afoo(\cZ)$ \Comment{$t$ is the sum of prediction weights.}
    \State $\cZ \gets \cZ \cdot \cX$ \Comment{Increase memory}
    \State $\evar \gets \max(1, \evar + (s - \hat{\escore})^2)$ ;\; $\hat{\escore} \gets \frac{1}{t}\cdot (t - 1) \cdot \hat{\escore} + s$
    \State \Return $[\hat{\escore} - \error_t(\evar, \conf, \sgn_{\loss}^{\afoo, \bfoo}), \hat{\escore} + \error_t(\evar, \conf,  \sgn_{\loss}^{\afoo, \bfoo})]$
\EndFunction
\end{algorithmic}
\end{algorithm}

\begin{remark}
    In formal verification we commonly assume the environment to be Markovian. Here it is sensible to define the weights as a function of the current state and current transition, i.e., $\afoo \colon \aS\to [0,c_{\afoo}]$ and $\bfoo\colon \aS\to (\aS\to [0,c_{\bfoo}])$. In this case, the monitor requires $\aO(1)$-space and $\aO(T_{\score, \aX})$-time per iteration.
\end{remark}

\section{Experiments}
\label{sec:experiments}
All experiments were run on an Apple M2 Pro with 16GB.

\paragraph{Expected scoring rules.}
We show the difference in behaviour between the two scoring rules, the Brier $\score_{B}$ and the spherical score $\score_{S}$ (see Example~\ref{ex:score_foo}), using a discretised and truncated Gaussian distribution over $100$ values. The environment distribution has mean $50$ and standard deviation $5$. The model uses the same distribution with different parameters.
Figure~\ref{fig:expected_score} depicts the changes to the score when modifying the mean or the standard deviation of the model. 
\begin{figure}[ht]
    \centering
     \includegraphics[width=1\linewidth]{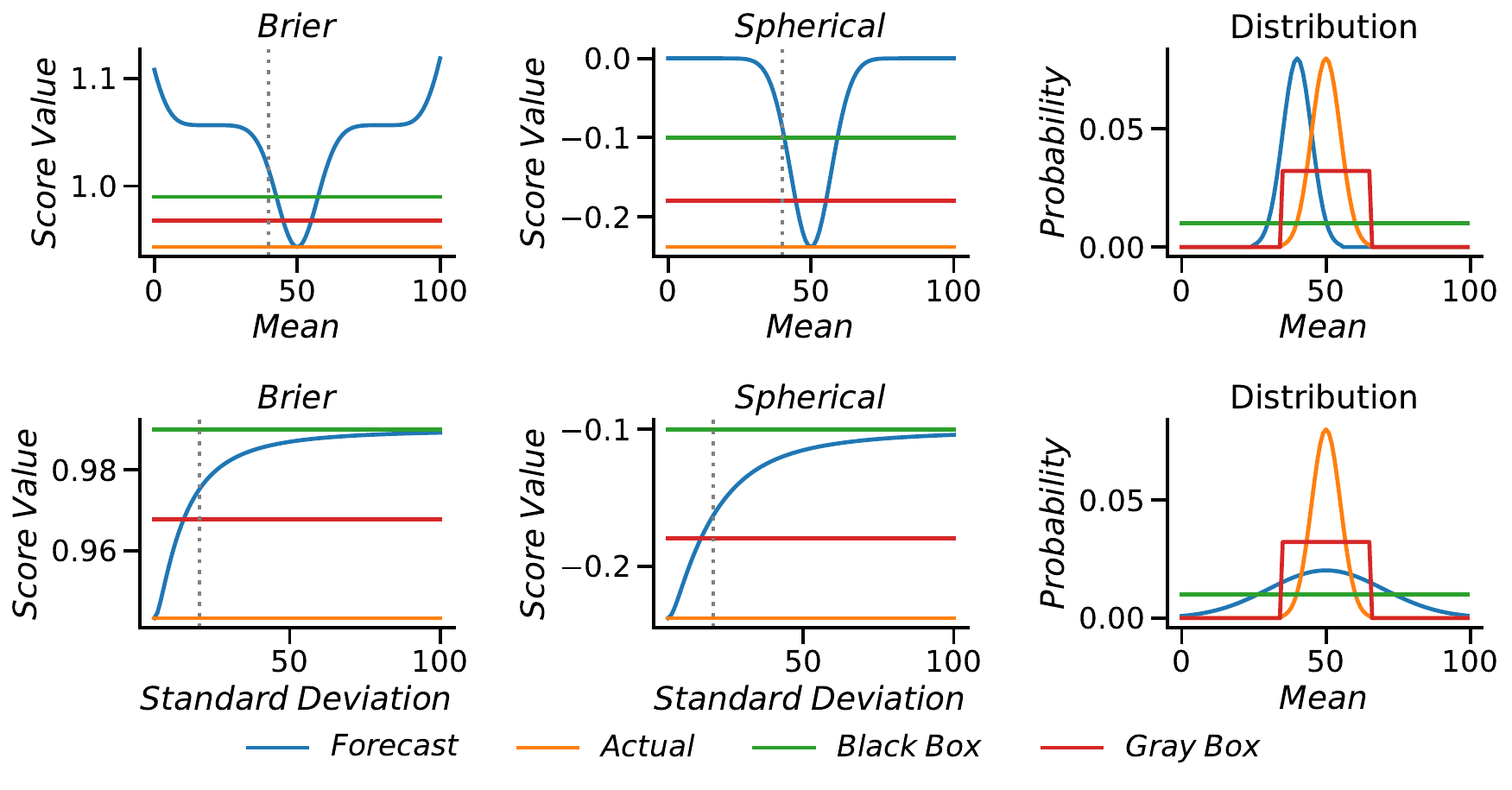}
    \caption{Behaviour of the expected Brier and spherical score: if the forecast mean is changed (Top); if the forecast standard deviation is changed (Bot). The gray line indicates the corresponding example in the right column.} 
    \label{fig:expected_score}
\end{figure}

\paragraph{Runtime.}
We evaluated the runtime of our monitor using the same discretised Gaussian distribution ranging from $10$ to $10^6$ values. Because our monitor is history-independent, we evaluated our monitor on a trace of length $10^4$ and averaged the runtime per iteration. We observe linear scaling and similar behaviour for both scoring rules. For example, for the Brier score and $10$ values, one iteration required on average $197\pm86\mu s$ of which $42\pm 16 \mu s$ is attributed to the scoring function. For $10^6$ values we have $259\pm 108\mu s$ and $72\pm 40 \mu s$ respectively.



\subsection{Average and Differential Alignment Monitoring}
We evaluate the average and the differential alignment monitor using the discrete-time Markov chains from the PRISM benchmark suite~\cite{kwiatkowska2012prism}.

\paragraph{Environments.}
We consider the transition matrices 
$\mathrm{Bench} = \{$Brp(16,2), Conditional, Crowds(5,5), Crowds(4,3), Die, Leader(3,5), Nand(5,2), Quantiles$\}$. To avoid bottom strongly connected components (BSCC), we add to each state a $0.01$ probability of returning to the initial state.

\paragraph{Models.}
We considered two forecast models obtained by corrupting each transition matrix $\ttrm$ in $\mathrm{Bench}$.
The corruption 
\emph{additive noise} $\etrm_{A}$ adds a scaled centred uniformly distributed random variable to the transition matrix, i.e., for every $\cS, \cS'\in \aS$ we have $\ttrm(\cS, \cS') + 0.1\cdot R$ where $R\sim \mathrm{Uniform}([-0.5, 0.5])$, and normalise each row.
The corruption \emph{invert} $\etrm_{I}$ computes the multiplicative inverse for every non-zero transition probability, i.e., for every $\cS, \cS'\in \aS$ we have $1/\ttrm(\cS, \cS') $ if $\trm(\cS, \cS')>0$, and normalise each row.

\paragraph{Reference models.}
We consider three reference models obtained from each Markov chain in 
$\mathrm{Bench}$. The black-box model $\etrm_{B}^{\rf}$, i.e., a uniform distribution over $\aS$, and the gray-box model $\etrm_{G}^{\rf}$, i.e., a uniform distribution over the successor states, as defined in Example~\ref{ex:mdp3}, an expert model $\etrm_{E}^{\rf}$, and the environment model $\ttrm$. The expert model is obtained by averaging the transition matrix of the environment and the gray-box model, i.e., $0.5 \cdot \ttrm + 0.5 \cdot \etrm_{G}$.

\paragraph{Average alignment monitor.}
For each environment $\ttrm\in \mathrm{Bench}$, each forecast model $\etrm\in \{\etrm_A, \etrm_I\}$, and each scoring function $\score\in \{ \score_B, \score_S\}$ we deploy our average alignment monitor once for $1000$ steps. Example runs for Crowds(4,3) are depicted in the first row of Figure~\ref{fig:prism}.
We add the average expected score for the reference models $\{\etrm_B^{\rf}, \etrm_G^{\rf}, \ttrm\}$ in order to place the estimated score in context. 
The values for $\etrm_B^{\rf}$ and $\etrm_G^{\rf}$ can always be computed.

\paragraph{Differential alignment monitor.}
For each environment $\ttrm\in \mathrm{Bench}$, each forecast model $\etrm\in \{\etrm_A, \etrm_I\}$, each reference model $\etrm^{\rf}\in \{\etrm_B^{\rf}, \etrm_G^{\rf}, \etrm_{E}^{\rf}, \ttrm\}$, and each scoring function $\score\in \{ \score_B, \score_S\}$ we deploy our differential alignment monitor $5$ times each for $1000$ steps. We record the first point in time where the monitor can make a decision. That is, if the monitor's upper bound is below $0$, then $\etrm$ is better aligned than $\etrm^{\rf}$; if the lower bound is above $0$, the inverse holds. We average the results and present them in Table~\ref{tab:prism_paper}. 
Example executions for Crowds(4,3) are depicted in the second row of Figure~\ref{fig:prism}.
The inversion of the probability distribution performed for $\etrm_I$ is a severe corruption of the transition probabilities. As a consequence, the monitor declares in most cases that $\trm^{\rf}$ is the better performing model, after only a few observations. By contrast, $\etrm_A$ is obtained by adding additive noise, which is a less severe corruption. As a consequence, it is less clear whether $\trm^{\rf}$ or $\etrm_A$ is better, which is also reflected in the need for more observations.
By construction, we know that the expert model is better aligned than the gray-box model, which is better aligned than the black-box model.
This is reflected in the decisions of the monitor, e.g., if $\etrm$ outperforms $\etrm_G^{\rf}$ then it must outperform $\etrm_B^{\rf}$.
The number of observations until a decision is a reflection of how difficult it is to distinguish the model from the reference model. 
\begin{figure}
    \centering
    \includegraphics[width=1\linewidth]{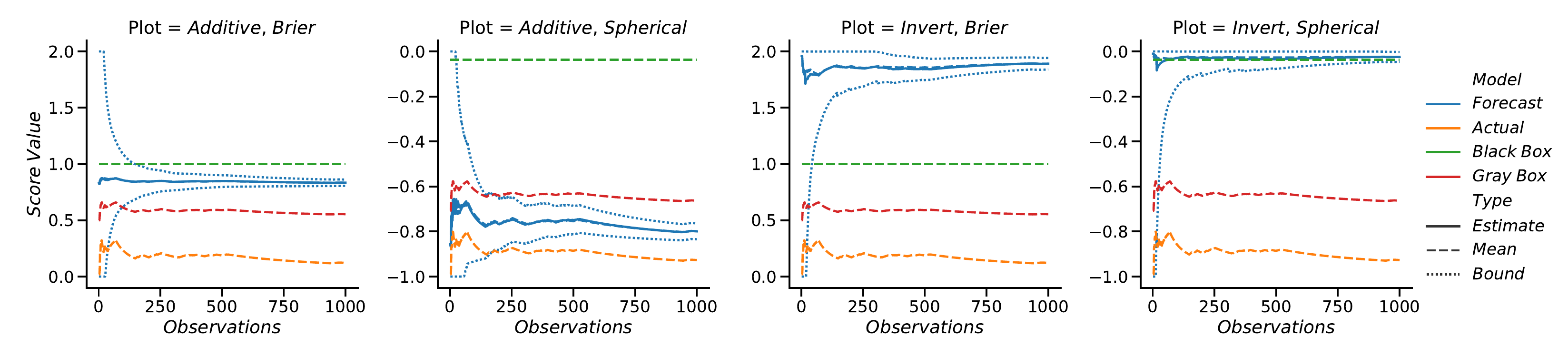}
    \includegraphics[width=1\linewidth]{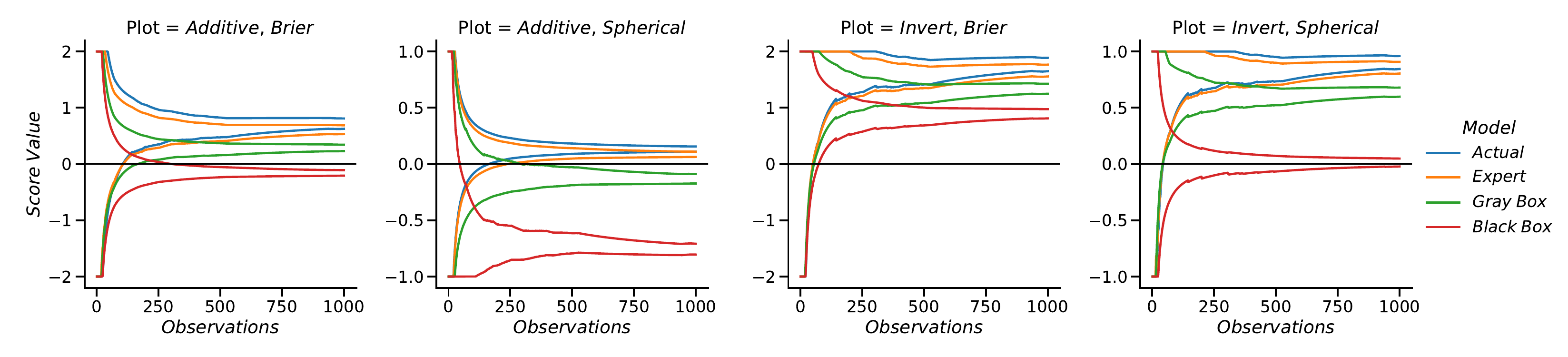}
    \caption{Example executions of the average alignment monitor (top) and the differential alignment monitor (bot) on the PRISM benchmark Crowds(4,3).}
    \label{fig:prism}
\end{figure}

\begin{table}[h]
    \centering
    \scriptsize
\begin{tabular}{llllll}
\toprule
 & Model & Environment & Expert & Gray  Box & Black  Box \\
Benchmark & Predictor &  &  &  &  \\
\midrule
\multirow[t]{2}{*}{Brp  (16,2)} & Additive & $70.6 \pm 17.17\; (\bot)$ & $72.8 \pm 8.56\; (\bot)$ & $175.2 \pm 2.49\; (\top)$ & $236.0 \pm 1.41\; (\top)$ \\
 & Invert & $39.0 \pm 5.79\; (\bot)$ & $40.4 \pm 4.16\; (\bot)$ & $52.0 \pm 2.83\; (\bot)$ & $78.0 \pm 9.46\; (\bot)$ \\
\multirow[t]{2}{*}{Conditional} & Additive & $1000.0 \pm 0.0\; (?)$ & $478.6 \pm 68.11\; (\top)$ & $113.0 \pm 23.45\; (\top)$ & $73.2 \pm 17.17\; (\top)$ \\
 & Invert & $28.6 \pm 5.81\; (\bot)$ & $29.4 \pm 5.37\; (\bot)$ & $33.4 \pm 5.37\; (\bot)$ & $40.6 \pm 5.81\; (\bot)$ \\
\multirow[t]{2}{*}{Crowds (5,5)} & Additive & $99.4 \pm 15.92\; (\bot)$ & $100.8 \pm 15.01\; (\bot)$ & $126.6 \pm 15.37\; (\bot)$ & $1000.0 \pm 0.0\; (?)$ \\
 & Invert & $51.6 \pm 4.22\; (\bot)$ & $51.2 \pm 3.63\; (\bot)$ & $54.0 \pm 4.36\; (\bot)$ & $68.0 \pm 15.52\; (\bot)$ \\
\multirow[t]{2}{*}{Die} & Additive & $1000.0 \pm 0.0\; (?)$ & $540.0 \pm 61.62\; (\top)$ & $120.2 \pm 26.35\; (\top)$ & $63.0 \pm 16.23\; (\top)$ \\
 & Invert & $31.0 \pm 5.61\; (\bot)$ & $31.8 \pm 5.17\; (\bot)$ & $34.4 \pm 5.37\; (\bot)$ & $50.0 \pm 8.22\; (\bot)$ \\
\multirow[t]{2}{*}{Leader (3,5)} & Additive & $84.8 \pm 18.09\; (\bot)$ & $94.8 \pm 4.15\; (\bot)$ & $351.4 \pm 3.44\; (\bot)$ & $125.0 \pm 1.73\; (\top)$ \\
 & Invert & $34.2 \pm 4.92\; (\bot)$ & $35.8 \pm 4.02\; (\bot)$ & $40.0 \pm 4.47\; (\bot)$ & $55.6 \pm 10.41\; (\bot)$ \\
\multirow[t]{2}{*}{Nand  (5,2)} & Additive & $84.4 \pm 9.94\; (\bot)$ & $79.8 \pm 9.36\; (\bot)$ & $130.8 \pm 2.59\; (\bot)$ & $561.6 \pm 10.6\; (\top)$ \\
 & Invert & $42.6 \pm 5.13\; (\bot)$ & $42.2 \pm 4.15\; (\bot)$ & $52.0 \pm 3.39\; (\bot)$ & $76.6 \pm 9.45\; (\bot)$ \\
\multirow[t]{2}{*}{Quantiles} & Additive & $1000.0 \pm 0.0\; (?)$ & $525.6 \pm 64.86\; (\top)$ & $128.8 \pm 33.57\; (\top)$ & $82.2 \pm 20.09\; (\top)$ \\
 & Invert & $40.2 \pm 9.31\; (\bot)$ & $42.4 \pm 9.79\; (\bot)$ & $47.8 \pm 10.35\; (\bot)$ & $61.0 \pm 14.88\; (\bot)$ \\
\multirow[t]{2}{*}{crowds-4-3} & Additive & $105.0 \pm 17.25\; (\bot)$ & $104.6 \pm 17.34\; (\bot)$ & $168.8 \pm 12.54\; (\bot)$ & $289.0 \pm 18.64\; (\top)$ \\
 & Invert & $46.4 \pm 3.65\; (\bot)$ & $44.6 \pm 4.72\; (\bot)$ & $48.8 \pm 4.27\; (\bot)$ & $63.6 \pm 14.66\; (\bot)$ \\
\bottomrule
\end{tabular}
    \caption{Average number of observations until decision. $\bot$ implies $\etrm^{\rf}$ is better than $\etrm$, $\top$ implies $\etrm^{\rf}$ is worse than $\etrm$, $?$ indicates indecision. A smaller number of observations is better, indicating an earlier differentiation of $\etrm^{\rf}$ and $\etrm$.}
    \label{tab:prism_paper}
\end{table}

\subsection{Weighted Alignment Monitoring}
We evaluate the weighted alignment monitor on Example~\ref{ex:fair} and Example~\ref{ex:safe}, using both the weighted Brier and the weighted spherical scoring rule, obtained through the transformation in Example~\ref{ex:weighted_rules}. 

\paragraph{Fairness.}
The environment transition matrix $\ttrm$ is taken from Example~\ref{ex:fair}, the model transition matrix $\etrm$ is obtained by flipping the transition probabilities of $S$ and the $G$ states, e.g., $\prob_{\etrm}(A\mid S)=0.2$ instead of $\prob_{\ttrm}(A\mid S)=0.8$. We use the alignment functions in Example~\ref{ex:weighted}, i.e., $\afoo(A)=\afoo(B)=1$ and $0$ otherwise.
In Figure~\ref{fig:weighted} we observe that the weighted alignment monitor does not distinguish the model from the environment, i.e., the alignment score computed w.r.t. the actual environment is contained in the interval. By contrast, the average alignment monitor clearly distinguishes the two. 
We notice that the bounds converge slower for the weighted monitor.

\paragraph{Safety.}
The environment transition matrix $\ttrm$ is taken from Example~\ref{ex:safe}, the model transition matrix $\etrm$ is obtained by attributing $0.1$ to each dotted transition, i.e., $\prob(s_5\mid s_6)=\prob(s_2\mid s_6)=0.1$ and $\prob(s_3\mid s_6)=0.8$. 
We use the alignment functions in Example~\ref{ex:weighted}, i.e., for all 
$s\in C \coloneqq \{s_3, s_5, s_6\}$ we have $\afoo(s)=1$ and $0.1$ otherwise; $\bfoo(s,s')=1$ if $c\in C$ and $s'\not\in C$, and $0.05$ otherwise. 
In Figure~\ref{fig:weighted} we observe that the weighted alignment monitor better distinguishes the model from the environment, i.e., the alignment score computed w.r.t.\ actual environment exits the interval earlier.  
We notice that the bounds converge slower for the weighted monitor.

\begin{figure}
    \centering
    \includegraphics[width=1\linewidth]{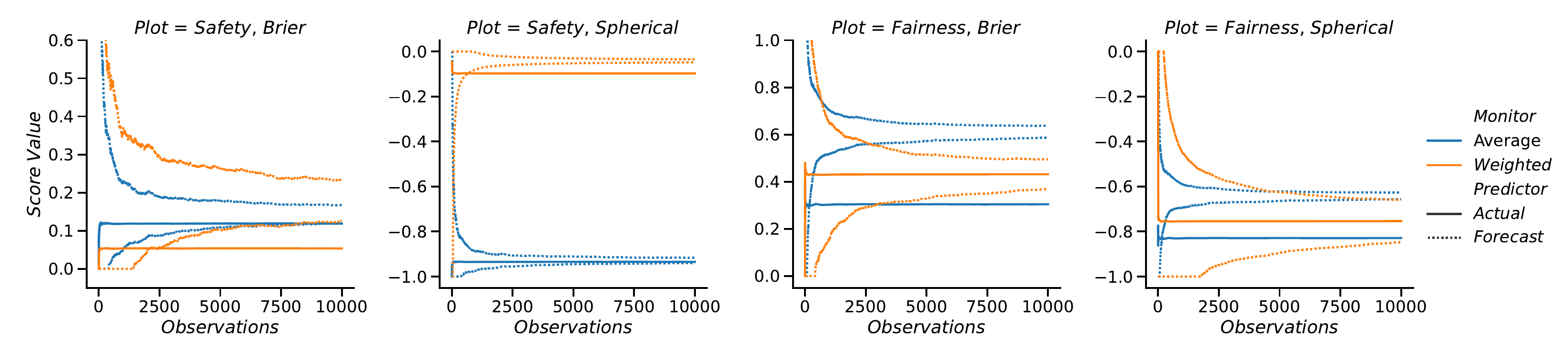}
    \caption{Example executions of the weighted alignment monitor on Example~\ref{ex:fair} (Fairness) and~\ref{ex:safe} (Safety).}. 
    \label{fig:weighted}
\end{figure}

\paragraph{Discussion.}
Both of the above examples demonstrate how weighted alignment monitors enable specification-specific alignment monitoring. In the fairness setting, the model and the actual environment are aligned on the two states relevant for the property value. Hence, considering the specification, the two models should not be distinguished. The weighted alignment monitor demonstrates this behaviour, because it focuses only on the property relevant states, while the general alignment monitor does not.  
In the safety setting, our system is safe as long as the BSCC in the model remains a BSCC in reality. 
The weighted alignment monitor emphasises the transitions away from this BSCC, and is therefore more sensitive to safety critical model misalignments.

\section{Related Work}
In this section we consider related work from verification and control, runtime methods, and (machine) learning.

\paragraph{Verification and control.}
Formal verification traditionally assumes that the model faithfully represents the real system. 
However, 
A recent line of work develops verification methods that can provide guarantees, even when the knowledge about the model is imprecise. This includes verification methods for robust MDPs, which are MDPs where the transition function is not known, but lies in an uncertainty set ~\cite{suilen2024robust}. This class of MDPs includes models such as bounded-parameter MDPs~\cite{givan2000bounded} or interval MDPs~\cite{jonsson1991specification,haddad2018interval}.
Robust MDPs can be defined by experts or learned from data allowing sound probably approximately correct verification~\cite{ashok2019pac}. Even if learned from data, there is no guarantee that the sampled data is representative of reality--especially if the data is obtained from a simulator.
In control theory, the alignment problem is known under model-plant mismatch, with works on quantifying model fidelity~\cite{badwe2010quantifying} and detecting a mismatch between model and reality~\cite{badwe2009detection}. In reinforcement learning it is known as the sim-to-real gap, which is concerned with bridging potential differences between reality and simulation~\cite{trentsios2022overcoming,stoffregen2020reducing}.

\paragraph{Runtime methods.}
There is work on monitoring the mismatch between the model and reality~\cite{halle2023leveraging,hosseinkhani2024model}. 
Both works address deterministic systems and detect misalignment using temporal logic specifications. 
In the verification of cyber-physical systems, the misalignment problem is also considered. For example Desai~\cite{desai2017combining} utilise signal temporal logic monitors to detect the violation of assumptions made during the model checking process. 
In reinforcement learning, approaches such as model-ensembles and runtime model falsification are utilized to overcome the sim-to-real gap at runtime~\cite{yamagata2020falsification,fulton2019verifiably}.
In our work, we are monitoring a hidden quantity, i.e., the expected score from realisations only. 
Hence, our monitor has to compute its verdict from partial observations. There is ample work on monitoring with imperfect or partial observations~\cite{ferrando2022runtime,junges2021runtime}.
In our work the hidden quantity is the expected value of a distribution, which can therefore be inferred using statistical methods. This is similar to a recent line of work focused on monitoring group fairness~\cite{henzinger2023runtime,henzinger2023monitoring}. 

\paragraph{Learning.}
A rich body of work addresses the problem of models facing changing data distributions or environments, often termed concept drift. Here they focus monitoring changes in a single distribution, e.g., a machine learning model is trained w.r.t.\ its training distribution. This distribution may shift over time \cite{hinder2023one,ditzler2015learning}. Formal verification models are usually stateful, with drastic differences in the successor distribution between time steps. This is why techniques from the sequential forecasting literature are more appropriate. 
In our work, we directly apply the techniques developed by Choe et al.~\cite{choe2021comparing} and Howard et al.~\cite{howard2021time} to build our monitors. Choe et al.~\cite{choe2021comparing} uses the confidence sequences developed by Howard et al.~\cite{howard2021time} to evaluate whether one forecaster outperforms another on average, i.e., our differential alignment monitoring problem. 
Together with Henzi et al.~\cite{henzi2022valid}, which develop a statistical test for detecting whether a forecaster outperforms another on every past prediction, they are the first to develop time-uniform statistical guarantees for forecaster evaluation. 

\paragraph{Property dependency.}
In formal verification, we are primarily interested in whether a system satisfies a given specification. Hence, we should be able to define alignment monitors with respect to a specification. Although weighted alignment monitors can adapt to a given property, this paper lacks a principled method for synthesizing a property-specific monitor. In developing such monitors, the literature on conformance testing may be particularly relevant. In conformance testing, the quality of a white-box model is assessed against a black-box model~\cite{roehm2019model}, potentially at runtime using monitors~\cite{mitsch2016modelplex}.
Model quality is evaluated using qualitative conformance relations~\cite{roehm2016reachset} or distance measures~\cite{abbas2014formal} on their output traces.
Most existing work, however, focuses on non-stochastic systems~\cite{roehm2019model}. In a stochastic setting, approaches akin to reward scaling may support the development of property-specific alignment monitors—particularly because, similar to our weighted alignment monitors, reward scaling emphasizes or penalizes behaviors in reinforcement learning to enforce formal specifications, such as safety~\cite{qian2023reward}.

\section{Conclusion}
The guarantees obtained by formal verification rely on the fidelity of the model. If the model is misaligned with reality, those guarantees can no longer be trusted. 
We introduced alignment monitoring as a runtime technique to assess whether a model used in formal verification remains aligned with the actual system behaviour. Our method leverages scoring rules and confidence sequences to track the predictive capabilities of a model over time. This requires no assumptions on the environment. We extended this with a differential and a weighted alignment monitor. The differential alignment monitor, compares the model against a trusted baseline, and the weighted alignment monitor, allows the flexibility to emphasise critical predictions. We evaluated our monitors on synthetic examples and on the PRISM benchmark suite. 

A clear extension is to develop more sophisticated alignment scores. This could include: utilising scoring rules defined over sequences of $k$ prediction in the construction of our average alignment monitor, as suggested in Choe et al.~\cite{choe2021comparing}; or the development of alignment scores tailored to formal verification applications. Another line of extension is to enrich the setting by considering alignment monitoring under partial observability.


\begin{credits}
\subsubsection*{\ackname} This work is supported by the European Research Council under Grant No.: ERC-2020-AdG 101020093.

\subsubsection*{\discintname}
The authors have no competing interests to declare that are
relevant to the content of this article. 
\end{credits}

\bibliographystyle{splncs04}
\bibliography{references}

\newpage

\appendix
\section{Appendix}
\subsection{Proofs}
Let $\rV$ and $\rW$ be the hidden and the observed process defined by the environment $\tenv$ and the model $\eenv$. Let $\rZ\coloneqq (\rX_t)_{t\in \pNN} $ be the stochastic process defined by $\tenv$. 
Let $\score_{\omega}\colon \sdistr(\aX)\times \aX \to [c_{\omega}\cdot a, c_{\omega} \cdot b]$ be a weighted scoring rule, let 
$\afoo\colon \aX^* \to [0,c_{\afoo}]$ be the prediction weight function, let $\bfoo\colon \aX^*\to (\aX \to [0,c_{\bfoo}])$ be the observation weight function, and $\conf\in (0,1)$ be a error probability threshold.
We define $t_{\afoo}(\cZ)\coloneqq \sum_{i=1}^t \afoo(\cZ_{1:i-1})$ as the weighted time, $\sgn_{\score}^{\afoo, \bfoo}\coloneqq c_{\afoo}c_{\bfoo}(b-a)$ as the weighted score bound, $\hat{\wscore}_t\coloneqq \hat{\wscore}_{\afoo, \bfoo}^{\score}(\rW_1, \dots, \rW_t)$ as the weighted average score process, and $\wscore_t\coloneqq \wscore_{\afoo, \bfoo}^{\score}(\rV_1, \dots, \rV_t)$ as the weighted average expected score process
\begin{lemma}
    \label{lemma:bound}
   The bound $\error_{t_{\afoo}}$, defined in Equation~\ref{eq:bound}, satisfies for every $\conf\in (0,1)$
    \begin{align*}
        \prob_{\tenv}(\exists t\in \pNN \colon |\hat{\wscore}_t - \wscore_t|\geq \error_{t_{\afoo}}(\evar_t,\conf, \sgn_{\score}) ) \leq \conf.
    \end{align*}
\end{lemma}
\begin{proof}
    We define two new stochastic processes $\hat{\rQ}\coloneqq (\hat{\rQ}_t)_{t\in\pNN}$ where
\begin{align*}
\hat{\rQ}_t\coloneqq \afoo(\rZ_{1:t-1})\cdot\score_{\bfoo(\rZ_{1:t-1})}(\reY_t, \rX_t)
\end{align*}
and $\rQ\coloneqq (\rQ_t)_{t\in\pNN}$ where
\begin{align*}
\rQ_t\coloneqq \afoo(\rZ_{1:t-1})\cdot \expe_{\rX_t\sim \rtY_t}(\score_{\bfoo(\rZ_{1:t-1})}(\reY_t, \rX_t)).
\end{align*}
Notice that $\rQ_t$ is almost surely bounded on $\sgn_{\score}^{\afoo, \bfoo}$.
We define the sum processes 
\begin{align*}
    \hat{M}\coloneqq (\hat{M}_t)_{t\in \pNN} = \left(\sum_{i=1}^t \hat{Z}_i\right)_{t\in \pNN} \quad \text{and} \quad M\coloneqq (M_t)_{t\in \pNN} = \left(\sum_{i=1}^t Z_i\right)_{t\in \pNN}
\end{align*}
which are the weighted average processes $\hat{\wscore}_t$ and $\wscore_t$ without the normalisation by $t_{\afoo}$. 
Moreover, we define the variance process as 
\begin{align*}
    \evar_t' \coloneqq   \sum_{i=1}^t (\hat{\rQ}_i - \hat{\wscore}_{i-1})^2.
\end{align*}
Notice that $\hat{\wscore}_{i-1}$ is a predictable process, i.e., intentively it uses only information before the time $i$. 
By Proposition 1 in Choe et al.~\cite{choe2021comparing}, we know that for every $t\in \pNN$ the process
\begin{align*}
 P_t\coloneqq \exp((\hat{M}_t-M_t) - \psi_{\sgn_{\score}^{\afoo,\bfoo}}(\lambda) \evar_t')
\end{align*}
is a test super martingale, where  $ \psi_{c}(\lambda) \coloneqq c^{-2}\left(-\log(1 - c\lambda) - c\lambda\right)$, which is defined for $\left[0, \frac{1}{c}\right)$. In our case $c=\sgn_{\score}^{\afoo,\bfoo}$.
Hence, we can directly apply the stitching bound from Howard et al.~\cite{howard2021time} to obtain for $\evar_t\coloneqq \max(1,\evar_t')$
\begin{align*}
    \mathrm{CI}_t(\evar_t,\conf, \sgn_{\score}^{\afoo,\bfoo}) \coloneqq &\sqrt{2.13 \cdot \evar_t \cdot g(\evar_t,\conf) + 1.76\cdot (\sgn_{\score}^{\afoo, \bfoo})^2 \cdot g(\evar_t,\conf)^2}
     \\
     &+ 1.33 \cdot \sgn_{\score}^{\afoo, \bfoo} \cdot g(\evar_t,\conf)
\end{align*}
where $h(n,\conf)=\left(2 \cdot \log\left(\pi\log(n) / \sqrt{6}\right) + \log(2/\conf)\right) $. The constants are obtained by setting $\eta=e$, $m=1$, and $h(x)\coloneqq \frac{6}{\pi^2} \frac{1}{x^2}$. The $2$ in $2/\conf$ is due to the union bound to get the two sided tail bound. We are guaranteed that 
\begin{align*}
    \prob_{\tenv}\left( \exists t\in \pNN \colon |\hat{M}_t-M_t |\geq \mathrm{CI}_t(\evar_t,\conf, \sgn_{\score}^{\afoo,\bfoo})\right) \leq \conf.
\end{align*}
We divide both sides by the predictable $t_{\afoo}$ to obtain
  \begin{align*}
        \prob_{\tenv}(\exists t\in \pNN \colon |\hat{\wscore}_{\afoo, \bfoo}^{\score}(\rW_1, \dots, \rW_t)- \wscore_{\afoo, \bfoo}^{\score}(\rV_1, \dots, \rV_t)|\geq \error_{t_{\afoo}}(\evar_t,\conf, \sgn_{\score}) ) \leq \conf.
\end{align*}
\end{proof}

\begin{proof}[Theorem~\ref{theorem:weighted_alignment}]
    The fact that monitor $\moni_{\score, \conf}$ satisfies the condition in Equation~\ref{eq:sound} for $\dist=\wscore_{\afoo, \bfoo}^{\score}$ is a direct consequence of Lemma~\ref{lemma:bound}. We need three counter, one for $t$, one for the average, and one for the variance.  In addition, we need to maintain a list storing the entire history of observations. This requires linear memory. Moreover, we require linear time to evaluate the weight functions. Moreover, evaluating the scoring rule is a function of $|\aX|$.
\end{proof}

\begin{proof}[Theorem~\ref{theorem:general_alignment}]
This is a special case of Theorem~\ref{theorem:weighted_alignment}, where $\afoo$ is constant $1$, $\score_{\omega}$ is weight independent, and $c_{\bfoo}=1$. 
Because of history independence, the monitor does not need to store the history. Hence, we obtain constant runtime and space requirement w.r.t. the history. Only evaluating the scoring rule is a function of $|\aX|$.
\end{proof}

\begin{proof}[Theorem~\ref{theorem:diff_alignment}]
    The fact that monitor $\moni_{\score, \conf}$ satisfies the condition in Equation~\ref{eq:diff_sound} for $\dist=\escore^{\score}$ is a direct consequence of Theorem 2 in Choe et al.~\cite{choe2021comparing}.
    We only need to adjust $\sgn$, because the minimum value is $a-b$ and the maximum value is $b-a$, resulting in $2(b-a)$.
    The space and time complexity follows from Theorem~\ref{theorem:general_alignment}, as the only difference is that we need to evaluate two scoring functions.
\end{proof}

\subsection{Experiments}
In Figure~\ref{fig:binary} we show the behaviour of the expected Brier and the spherical score for a binary predictor for varying true probability values.
The black box predictor is set to $0.5$.
In Figure~\ref{fig:binary_moni_35} and ~\ref{fig:binary_moni_90} we depict example executions of the average alignment monitor for a true probability value of $0.35$ and $0.9$ respectively. The predicted probability values are top to bottom $0.1$, $0.6$, and $0.8$

If Figure~\ref{fig:expected_monitor} we depict example executions of the alignment monitor for the distributions in Figure~\ref{fig:expected_score}.
The left plot corresponds to the distribution on the bottom right. The right plot corresponds to the distribution on the top right. We can observe the effect of a larger standard deviation on the convergence speed of our monitor.  

We transform the environment matrices in the PRISM benchmark suit using various corruptions. We list them roughly below. 
\begin{itemize}
  \item $P_{\mathrm{sharp}}$: Each row of $P$ is sharpened by raising its entries to the power $4.0$, then normalised. High probabilities are amplified.

  \item $P_{\mathrm{supp}}$: Each row is resampled to have a new support, retaining one original non-zero entry and sampling the rest randomly. Default keep probability is $0.3$.

  \item $P_{\mathrm{noisy}}$: Uniform noise scaled by $0.1$ is added to each entry of $P$, clipped to non-negative values, and rows are renormalised.

  \item $P_{\mathrm{drop}}$: Entries in each row are randomly dropped with probability $0.4$, ensuring at least one non-zero value ($\text{min\_support} = 1$). Remaining mass is renormalised.

  \item $P_{\mathrm{swap}}$: For each row, the maximum and minimum entries are swapped, keeping the row normalised.

  \item $P_{\mathrm{collapse}}$: Each row is collapsed to concentrate nearly all mass on one randomly chosen entry, producing near-deterministic behaviour.

  \item $P_{\mathrm{bias}}$: Each row is biased towards a fixed target index ($\text{target} = 0$) by convex combination with a one-hot vector; bias strength is $0.55$.

  \item $P_{\mathrm{inv}}$: Each row’s non-zero entries are inverted (replaced by their reciprocals) and then renormalised, so low probabilities become dominant.

  \item $P_{\mathrm{flip}}$: For each non-zero entry $p$ in a row, it is replaced by $1 - p$ and the row is renormalised. Falls back to the original row if the sum vanishes.
\end{itemize}
In Figure~\ref{fig:average_prism_1} and~\ref{fig:average_prism_2} we depict example executions of the alignment monitor for each of the above corruptions and each PRISM benchmark, for both the Brier and the spherical score. 
Figure~\ref{fig:diff_prism_1} and~\ref{fig:diff_prism_1} shows the same for the differential alignment monitor, with the reference models constructed as in Section~\ref{sec:experiments}. Table~\ref{tab:term_1} and~\ref{tab:term_2} shows the average time at which the monitor makes a conclusive decision computed over 5 runs each 1000 steps.

\begin{figure}
    \centering
    \includegraphics[width=1\linewidth]{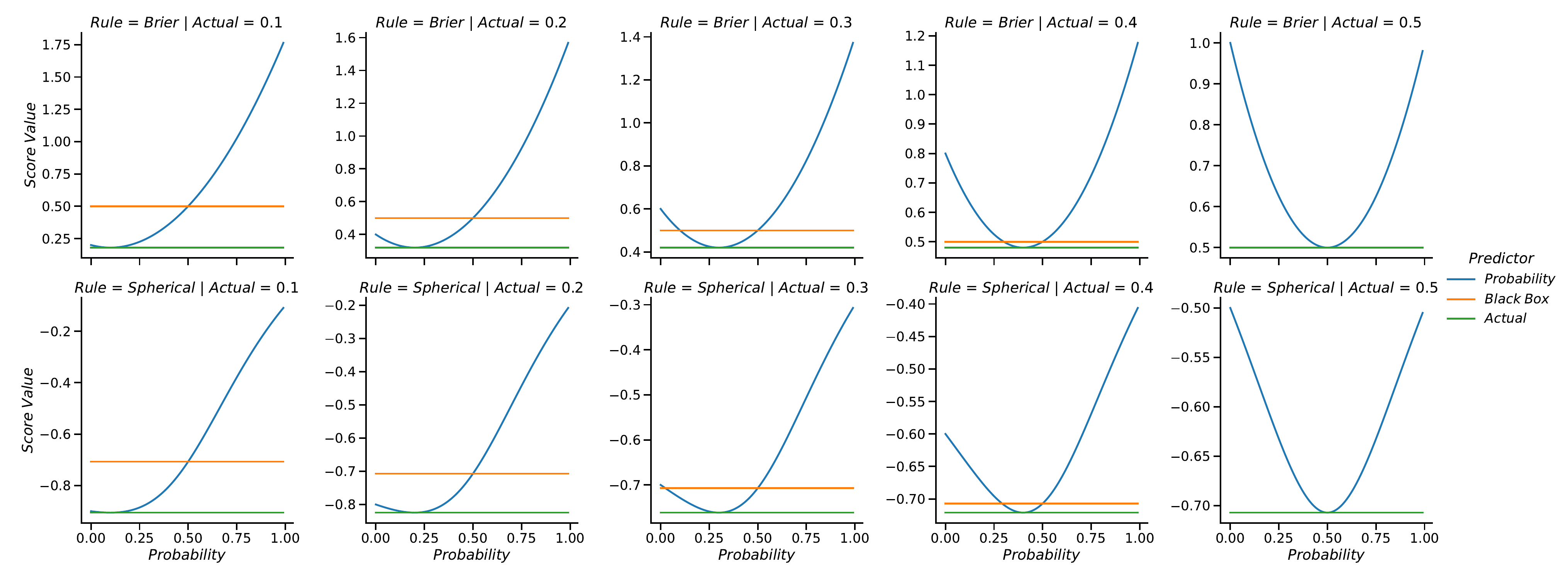}
    \caption{Behaviour of the expected Brier (top) and the spherical score (bot) for a binary predictor for varying true probability values}
    \label{fig:binary}
\end{figure}

\begin{figure}
    \centering
    \includegraphics[width=1\linewidth]{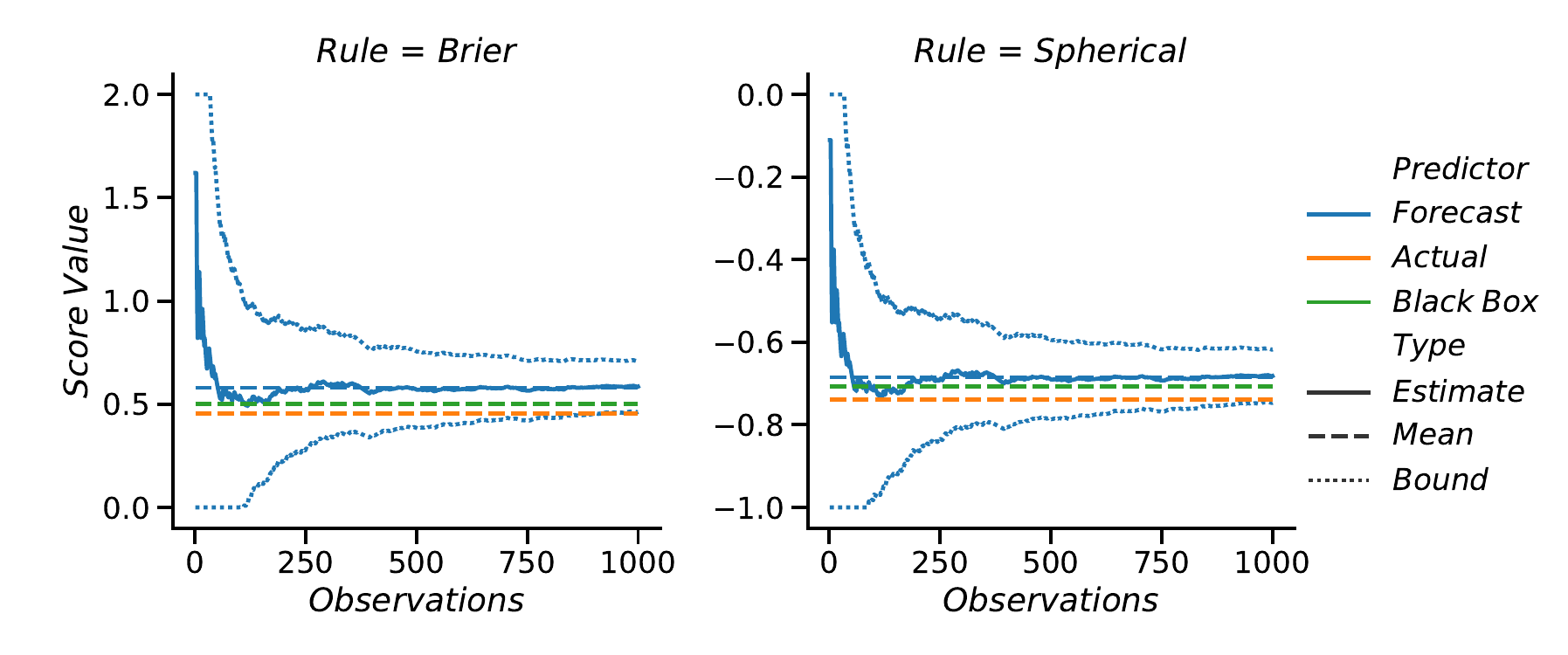}
    \includegraphics[width=1\linewidth]{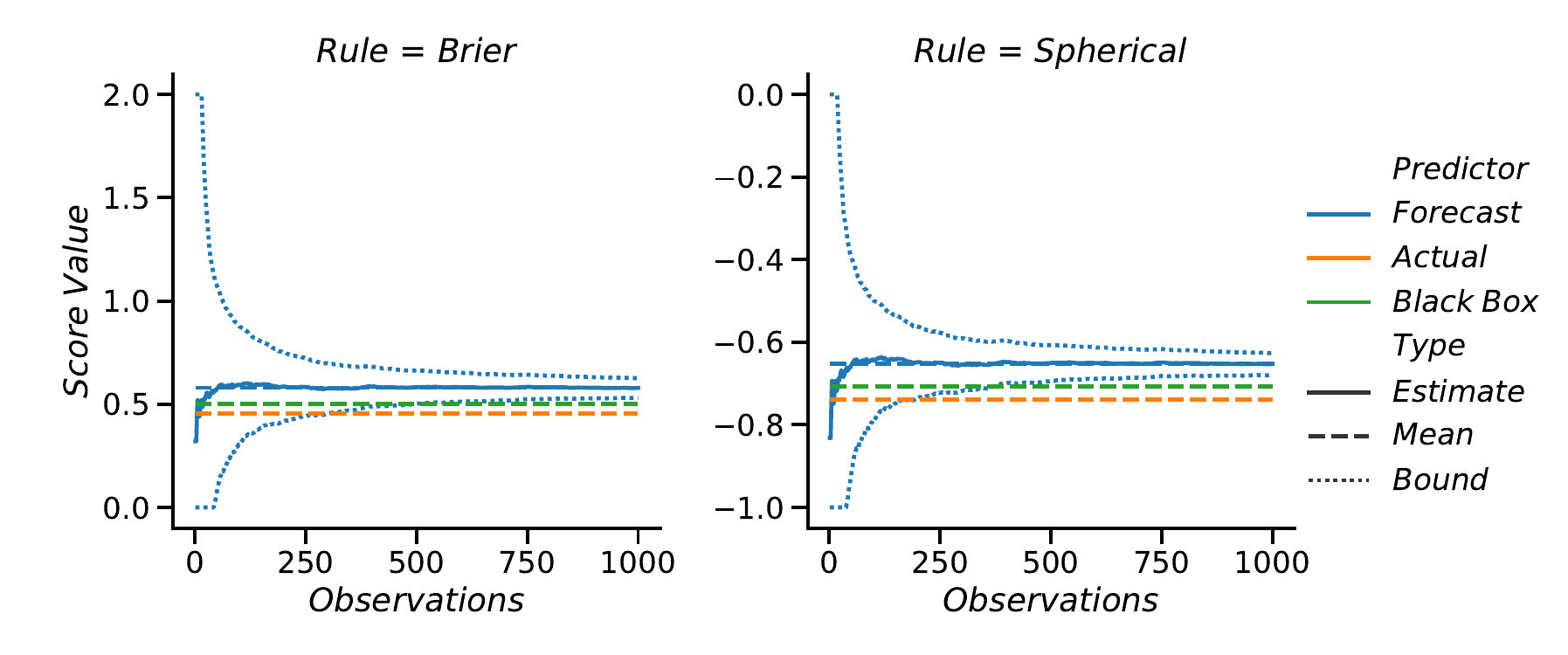}
    \includegraphics[width=1\linewidth]{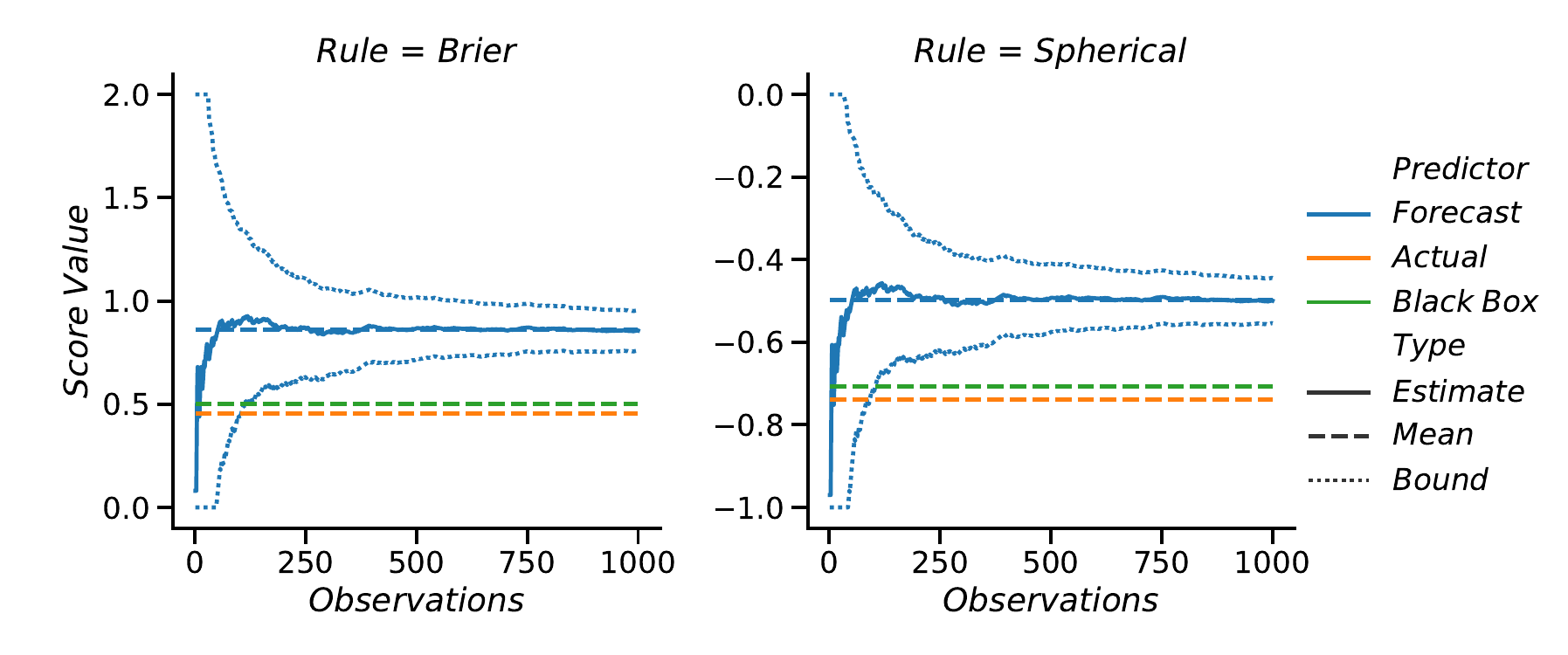}
    \caption{Example executions of the average alignment monitor for a true probability value of $0.35$. The predicted probability values are top to bottom $0.1$, $0.6$, and $0.8$.}
    \label{fig:binary_moni_35}
\end{figure}

\begin{figure}
    \centering
    \includegraphics[width=1\linewidth]{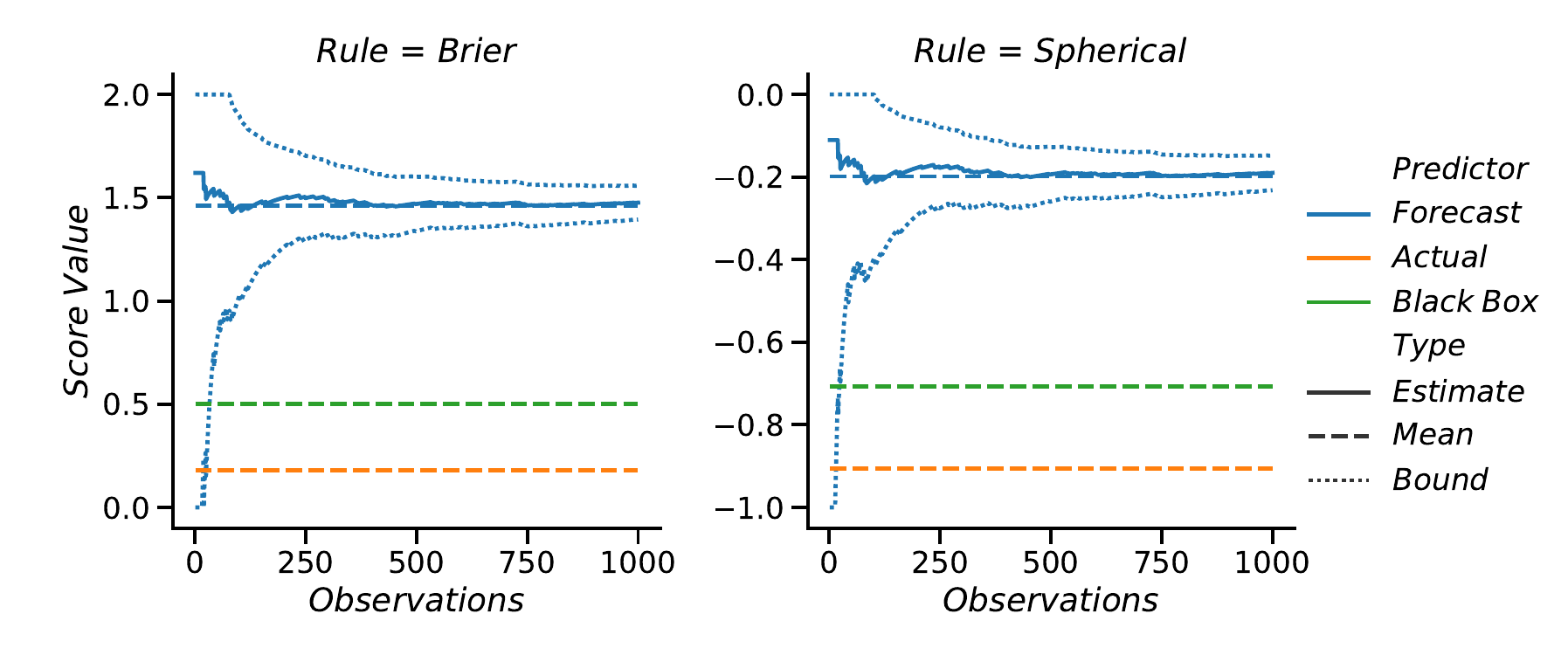}
    \includegraphics[width=1\linewidth]{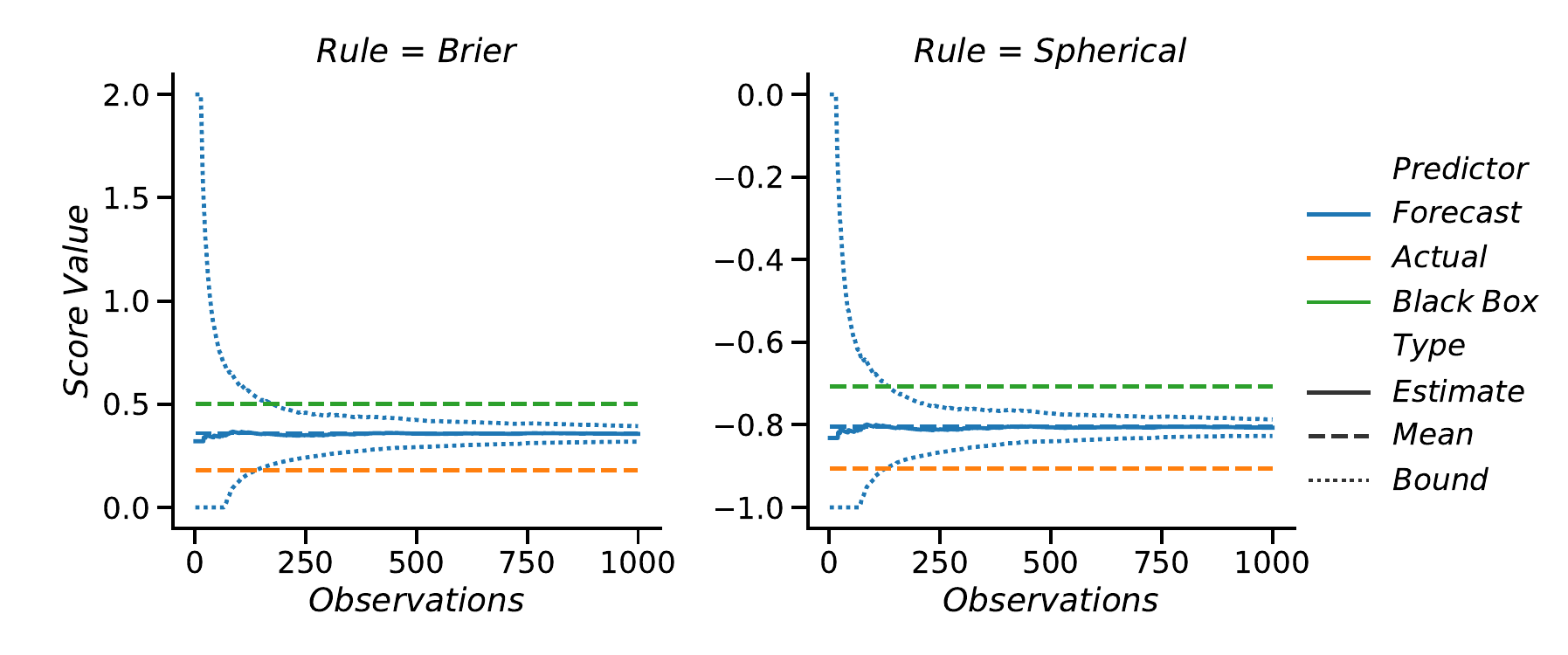}
    \includegraphics[width=1\linewidth]{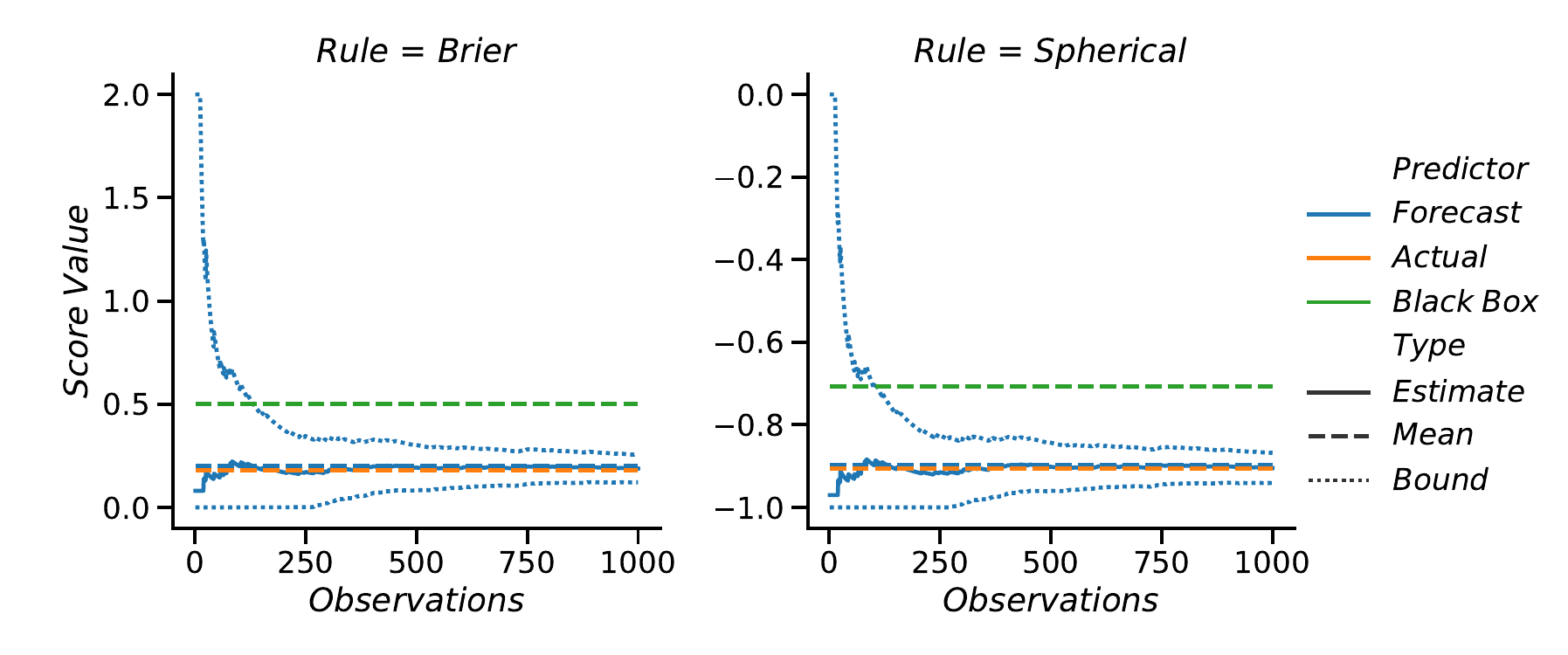}
    \caption{Example executions of the average alignment monitor for a true probability value of $0.9$. The predicted probability values are top to bottom $0.1$, $0.6$, and $0.8$.}
    \label{fig:binary_moni_90}
\end{figure}


\begin{figure}[ht]
    \centering
    \includegraphics[width=0.9\linewidth]{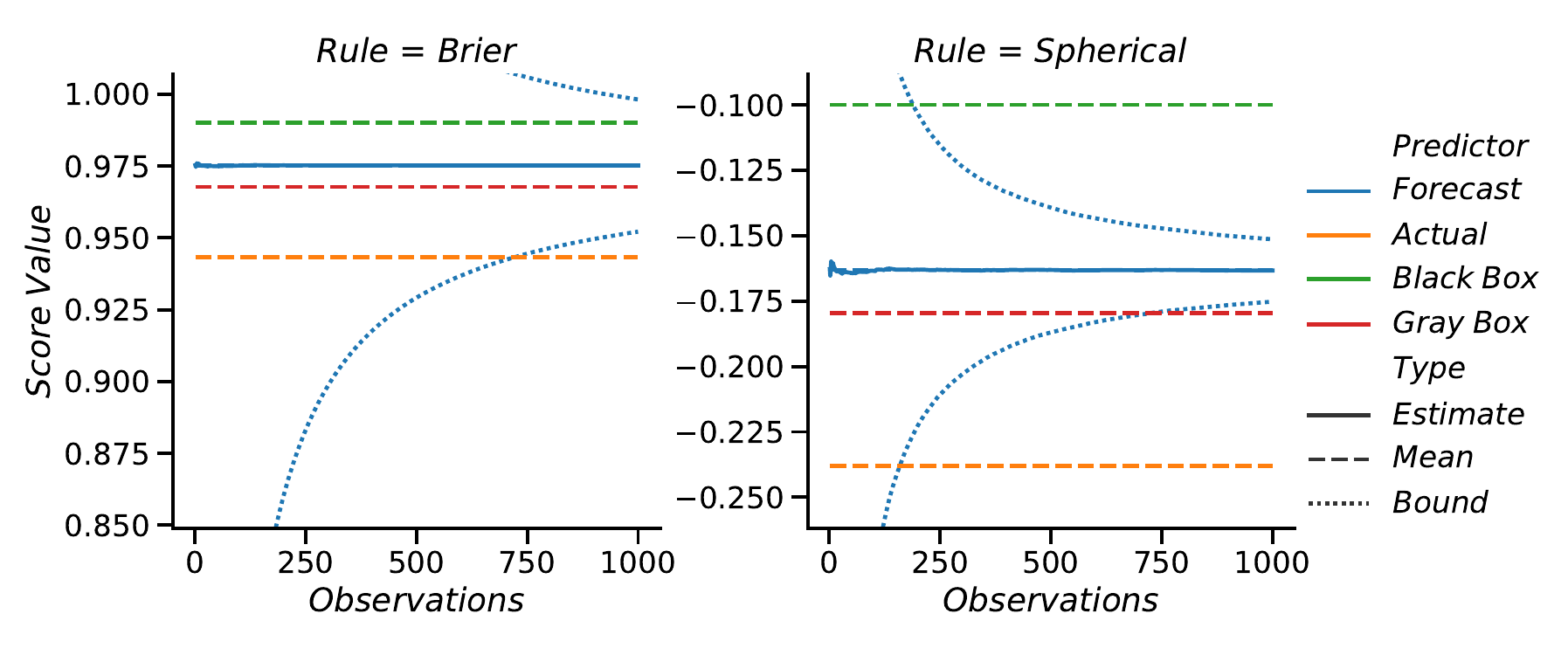}
    \caption{Example executions of the alignment monitor for the distributions in Figure~\ref{fig:expected_score}.
    The left plot corresponds to the distribution on the bottom right. The right plot corresponds to the distribution on the top right.}
    \label{fig:expected_monitor}
\end{figure}


\begin{sidewaysfigure}
    \centering
    \begin{tabular}{@{}p{0.2cm} c@{}}
        \raisebox{1\height}{\rotatebox{90}{brp-16-2}} & \includegraphics[width=0.9\linewidth]{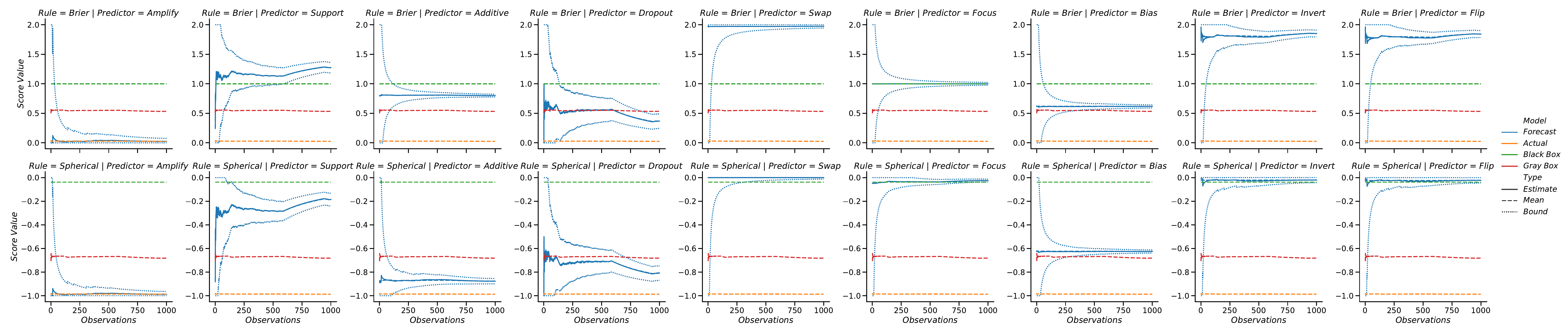} \\
        \raisebox{1\height}{\rotatebox{90}{crowds-4-3}} & \includegraphics[width=0.9\linewidth]{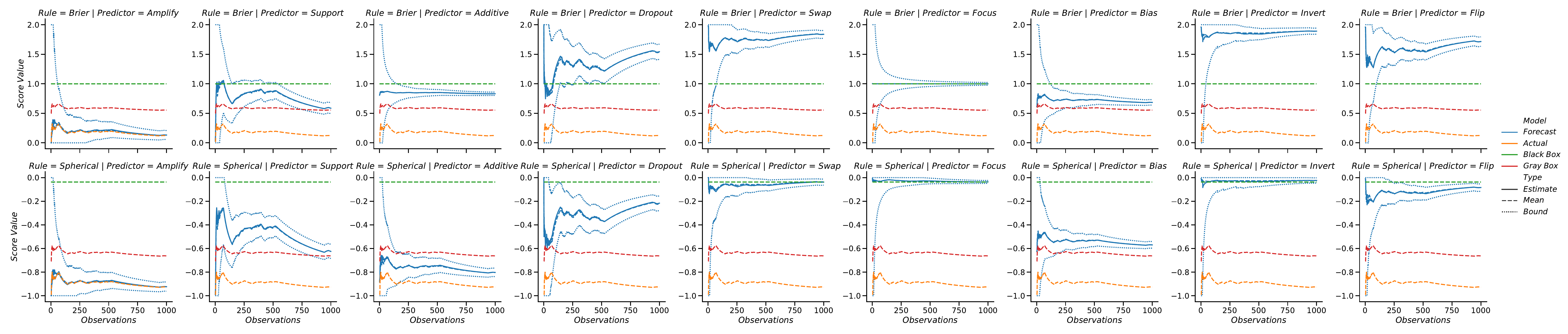} \\
        \raisebox{1\height}{\rotatebox{90}{crowds-5-5}} & \includegraphics[width=0.9\linewidth]{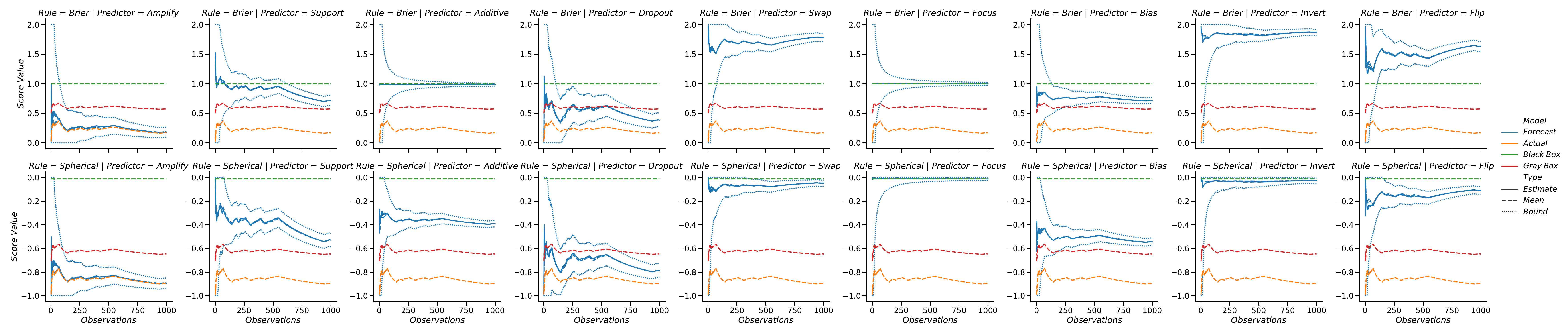} \\
        \raisebox{2\height}{\rotatebox{90}{die}} & \includegraphics[width=0.9\linewidth]{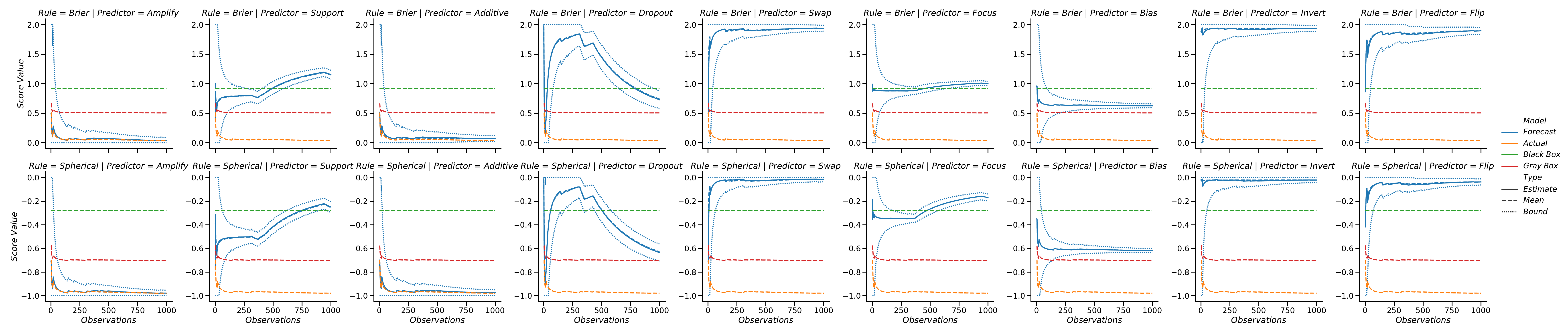} \\
    \end{tabular}
    \caption{Example executions of the average alignment monitor.}
    \label{fig:average_prism_1}
\end{sidewaysfigure}

\begin{sidewaysfigure}
    \centering
    \begin{tabular}{@{}p{0.2cm} c@{}}
        \raisebox{1\height}{\rotatebox{90}{leader-3-5}} & \includegraphics[width=0.9\linewidth]{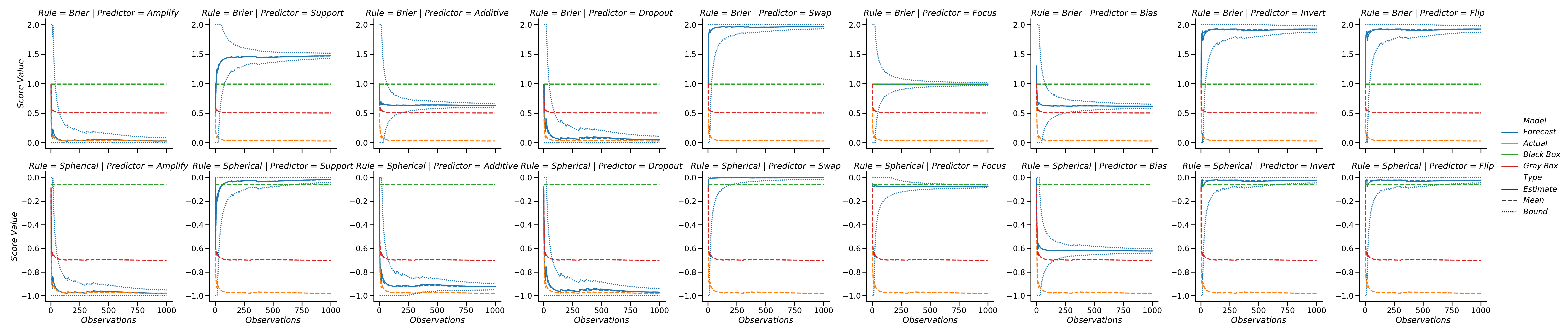} \\
        \raisebox{1\height}{\rotatebox{90}{nand-5-2}} & \includegraphics[width=0.9\linewidth]{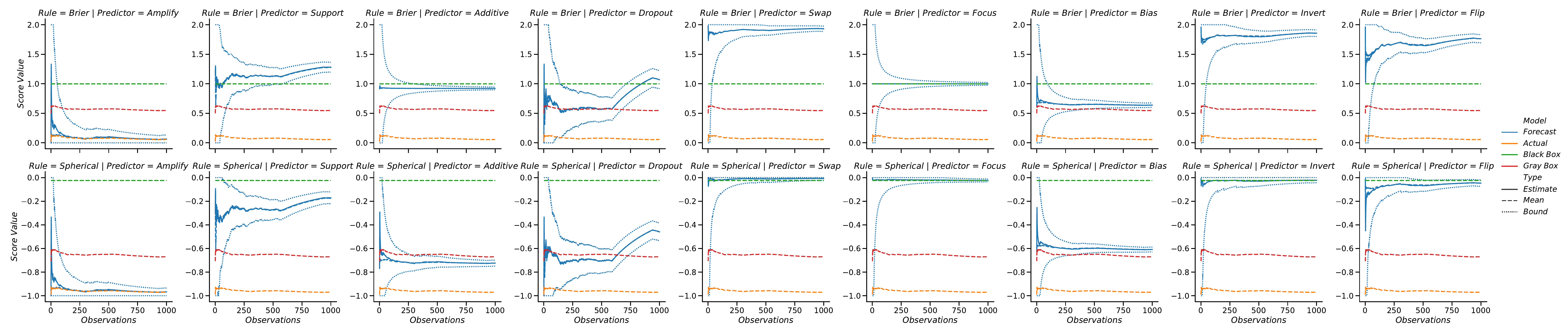} \\
        \raisebox{1\height}{\rotatebox{90}{quantiles}} & \includegraphics[width=0.9\linewidth]{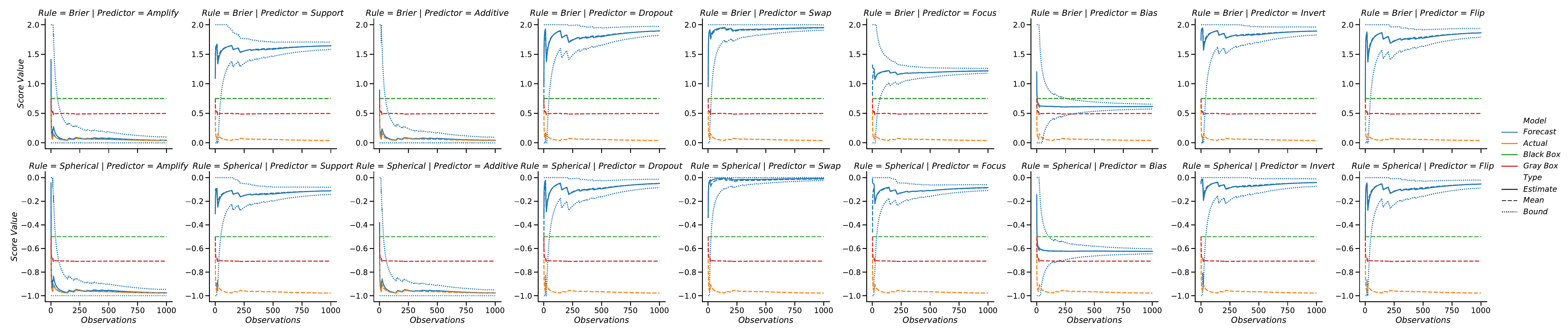} \\
        \raisebox{1\height}{\rotatebox{90}{test-cond}} & \includegraphics[width=0.9\linewidth]{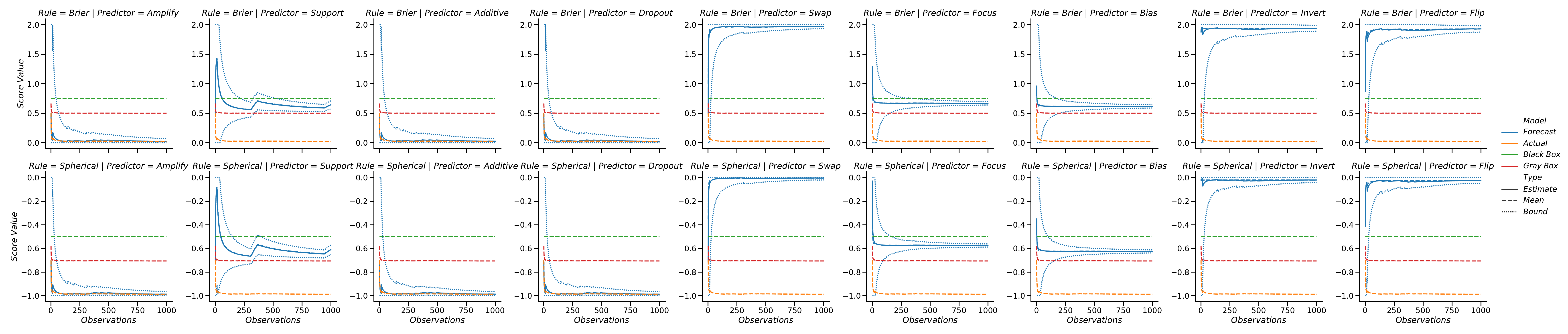}
    \end{tabular}
    \caption{Example executions of the average alignment monitor.}
    \label{fig:average_prism_2}
\end{sidewaysfigure}

\begin{sidewaysfigure}
    \centering
    \begin{tabular}{@{}p{0.2cm} c@{}}
        \raisebox{1\height}{\rotatebox{90}{brp-16-2}} & \includegraphics[width=0.9\linewidth]{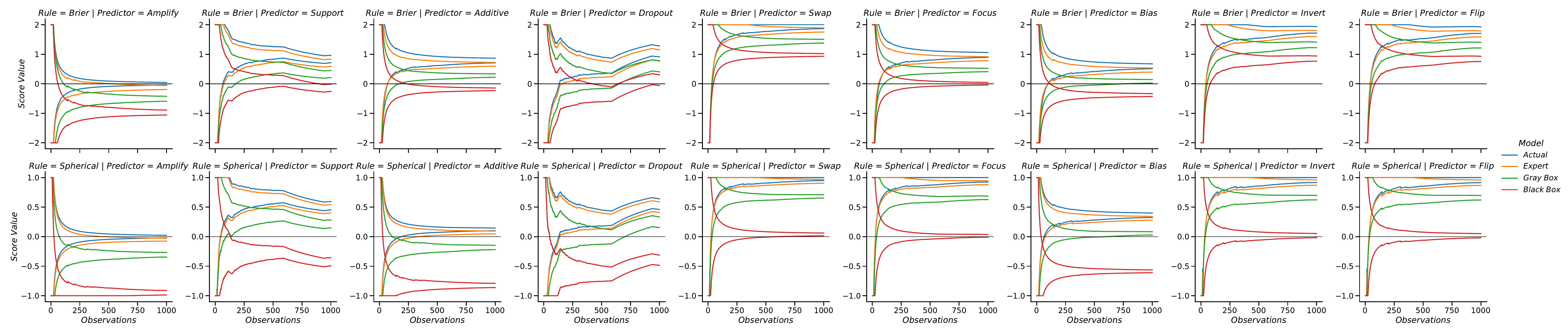} \\
        \raisebox{1\height}{\rotatebox{90}{crowds-4-3}} & \includegraphics[width=0.9\linewidth]{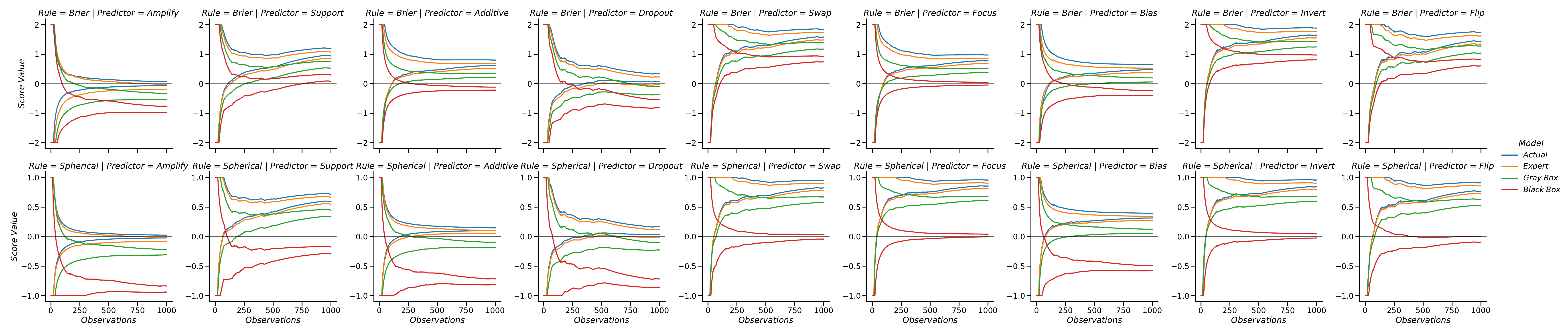} \\
        \raisebox{1\height}{\rotatebox{90}{crowds-5-5}} & \includegraphics[width=0.9\linewidth]{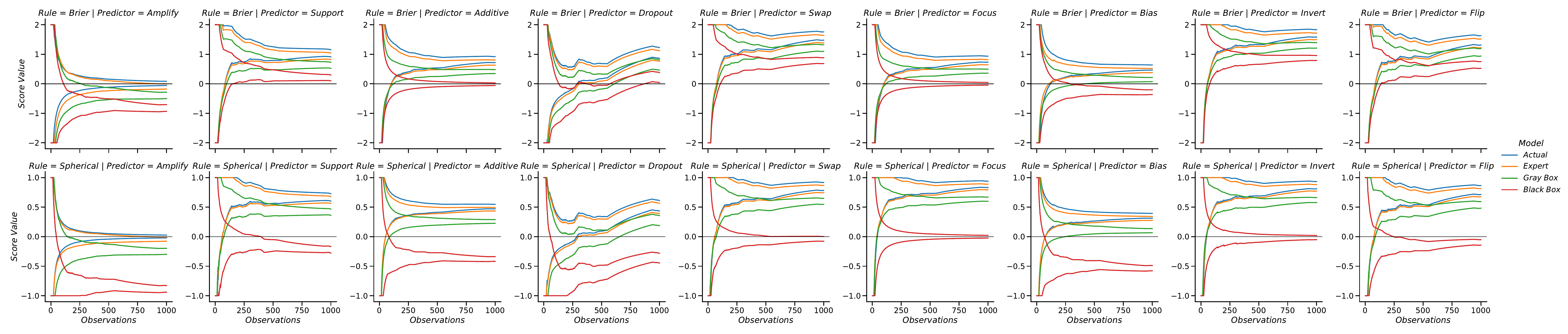} \\
        \raisebox{2\height}{\rotatebox{90}{die}} & \includegraphics[width=0.9\linewidth]{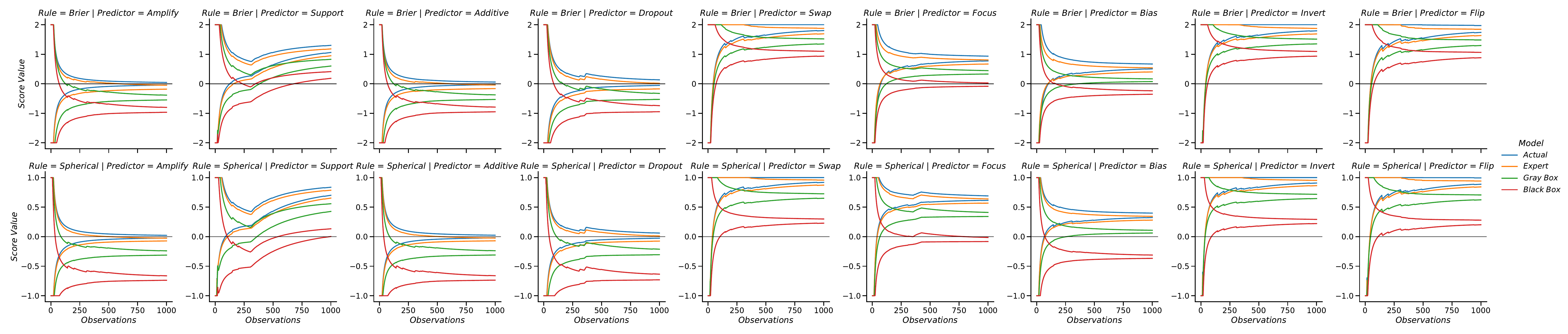} \\
    \end{tabular}
    \caption{Example executions of the differential alignment monitor.}
    \label{fig:diff_prism_1}
\end{sidewaysfigure}

\begin{sidewaysfigure}
    \centering
    \begin{tabular}{@{}p{0.2cm} c@{}}
        \raisebox{1\height}{\rotatebox{90}{leader-3-5}} & \includegraphics[width=0.9\linewidth]{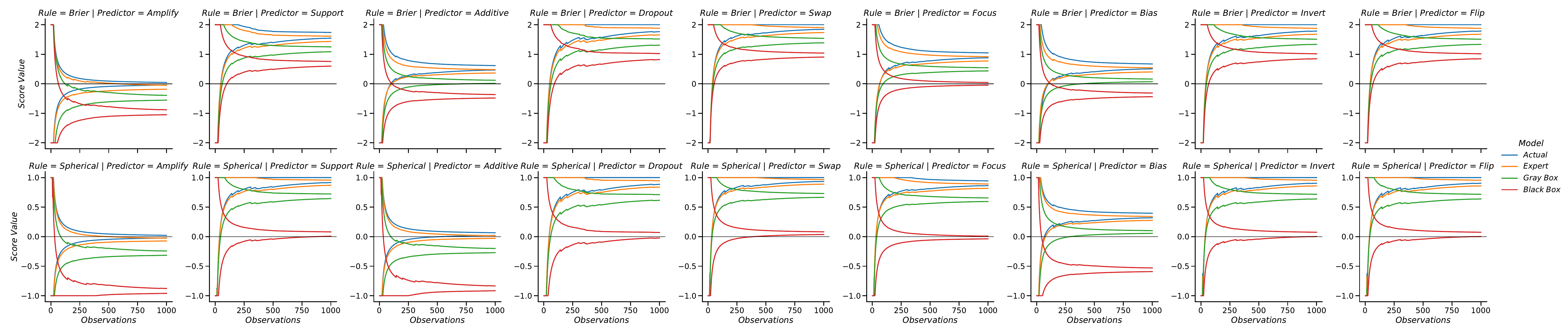} \\
        \raisebox{1\height}{\rotatebox{90}{nand-5-2}} & \includegraphics[width=0.9\linewidth]{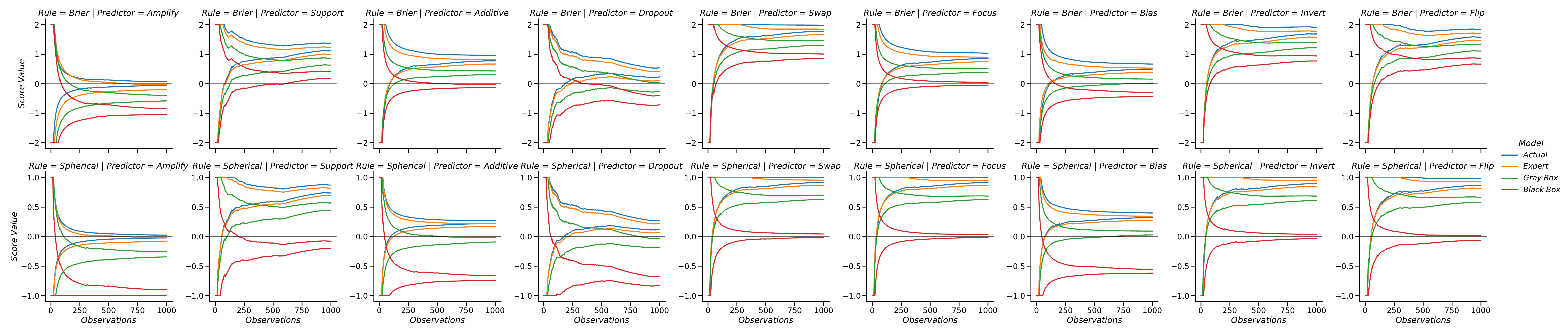} \\
        \raisebox{1\height}{\rotatebox{90}{quantiles}} & \includegraphics[width=0.9\linewidth]{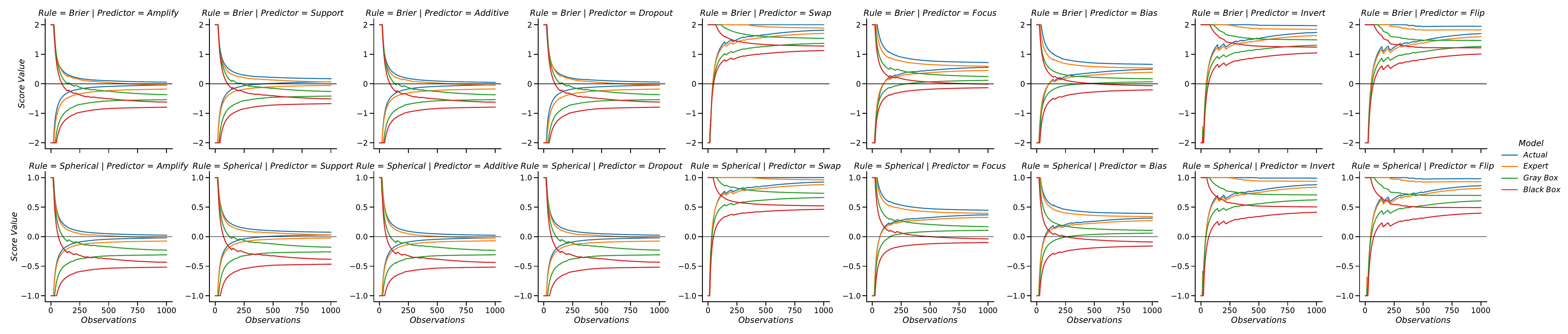} \\
        \raisebox{1\height}{\rotatebox{90}{test-cond}} & \includegraphics[width=0.9\linewidth]{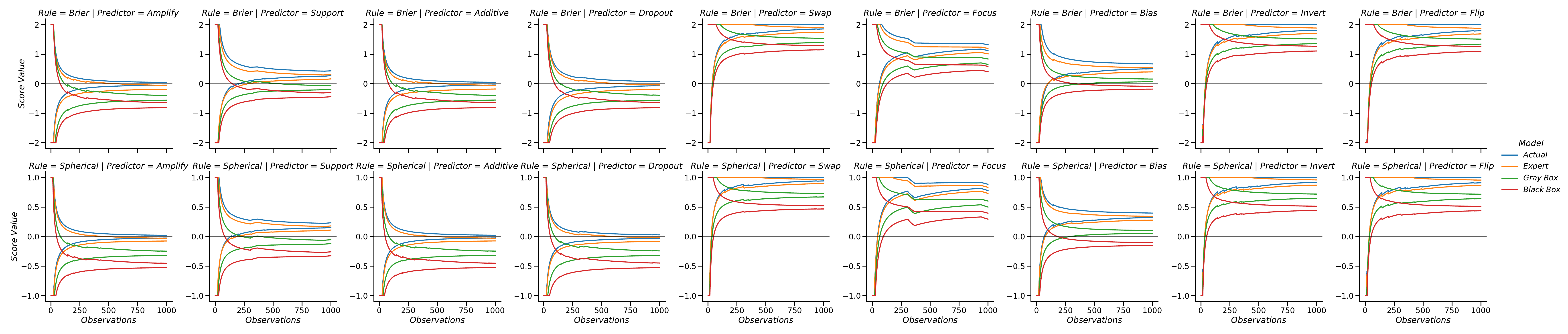}
    \end{tabular}
    \caption{Example executions of the differential alignment monitor.}
    \label{fig:diff_prism_2}
\end{sidewaysfigure}

\newpage

\begin{table}[p]
    \small
    \begin{tabular}{llllll}
\toprule
 & Model & Actual & Black  Box & Expert & Gray  Box \\
Benchmark & Predictor &  &  &  &  \\
\midrule
\multirow[t]{9}{*}{Brp  (16,2)} & Additive & $70.6 \pm 17.17 $ & $236.0 \pm 1.41 $ & $72.8 \pm 8.56 $ & $175.2 \pm 2.49 $ \\
 & Amplify & $1000.0 \pm 0.0 $ & $57.2 \pm 12.34 $ & $491.0 \pm 45.51 $ & $117.4 \pm 24.76 $ \\
 & Bias & $89.6 \pm 24.12 $ & $119.4 \pm 2.51 $ & $103.8 \pm 14.75 $ & $706.8 \pm 71.3 $ \\
 & Dropout & $135.6 \pm 17.42 $ & $462.2 \pm 151.6 $ & $164.2 \pm 26.58 $ & $942.0 \pm 129.69 $ \\
 & Flip & $39.8 \pm 5.72 $ & $83.0 \pm 9.67 $ & $41.4 \pm 4.39 $ & $53.2 \pm 3.11 $ \\
 & Focus & $55.8 \pm 10.33 $ & $1000.0 \pm 0.0 $ & $55.6 \pm 3.71 $ & $101.2 \pm 0.45 $ \\
 & Invert & $39.0 \pm 5.79 $ & $78.0 \pm 9.46 $ & $40.4 \pm 4.16 $ & $52.0 \pm 2.83 $ \\
 & Support & $59.0 \pm 4.64 $ & $281.6 \pm 40.67 $ & $66.4 \pm 3.78 $ & $104.6 \pm 6.47 $ \\
 & Swap & $30.6 \pm 2.19 $ & $47.0 \pm 0.0 $ & $30.8 \pm 1.1 $ & $36.0 \pm 0.0 $ \\
\cline{1-6}
\multirow[t]{9}{*}{Conditional} & Additive & $1000.0 \pm 0.0 $ & $73.2 \pm 17.17 $ & $478.6 \pm 68.11 $ & $113.0 \pm 23.45 $ \\
 & Amplify & $1000.0 \pm 0.0 $ & $73.2 \pm 18.01 $ & $492.8 \pm 67.61 $ & $114.2 \pm 24.57 $ \\
 & Bias & $89.2 \pm 22.48 $ & $338.6 \pm 1.95 $ & $100.8 \pm 14.17 $ & $390.8 \pm 0.84 $ \\
 & Dropout & $379.6 \pm 365.53 $ & $123.0 \pm 62.66 $ & $515.6 \pm 452.18 $ & $182.6 \pm 104.34 $ \\
 & Flip & $32.2 \pm 4.92 $ & $46.0 \pm 8.22 $ & $33.0 \pm 4.47 $ & $38.2 \pm 4.92 $ \\
 & Focus & $98.6 \pm 10.81 $ & $240.4 \pm 38.29 $ & $122.0 \pm 2.12 $ & $1000.0 \pm 0.0 $ \\
 & Invert & $28.6 \pm 5.81 $ & $40.6 \pm 5.81 $ & $29.4 \pm 5.37 $ & $33.4 \pm 5.37 $ \\
 & Support & $64.8 \pm 59.78 $ & $109.0 \pm 48.74 $ & $80.8 \pm 89.58 $ & $135.0 \pm 142.26 $ \\
 & Swap & $29.2 \pm 4.6 $ & $40.8 \pm 5.22 $ & $30.2 \pm 4.44 $ & $34.8 \pm 5.22 $ \\
\cline{1-6}
\multirow[t]{9}{*}{Crowds (5,5)} & Additive & $99.4 \pm 15.92 $ & $1000.0 \pm 0.0 $ & $100.8 \pm 15.01 $ & $126.6 \pm 15.37 $ \\
 & Amplify & $1000.0 \pm 0.0 $ & $109.0 \pm 18.81 $ & $885.0 \pm 128.43 $ & $244.2 \pm 26.58 $ \\
 & Bias & $115.8 \pm 21.51 $ & $271.8 \pm 23.69 $ & $110.6 \pm 18.11 $ & $428.4 \pm 17.98 $ \\
 & Dropout & $414.6 \pm 355.78 $ & $283.6 \pm 120.96 $ & $459.2 \pm 344.65 $ & $456.0 \pm 280.27 $ \\
 & Flip & $80.8 \pm 18.83 $ & $165.4 \pm 50.14 $ & $82.6 \pm 19.31 $ & $96.6 \pm 25.6 $ \\
 & Focus & $98.2 \pm 15.97 $ & $1000.0 \pm 0.0 $ & $99.8 \pm 15.01 $ & $124.2 \pm 16.27 $ \\
 & Invert & $51.6 \pm 4.22 $ & $68.0 \pm 15.52 $ & $51.2 \pm 3.63 $ & $54.0 \pm 4.36 $ \\
 & Support & $102.2 \pm 48.75 $ & $541.4 \pm 268.04 $ & $107.0 \pm 53.31 $ & $205.2 \pm 159.16 $ \\
 & Swap & $66.2 \pm 8.84 $ & $107.0 \pm 20.53 $ & $66.6 \pm 8.65 $ & $75.8 \pm 11.12 $ \\
\cline{1-6}
\multirow[t]{9}{*}{Die} & Additive & $1000.0 \pm 0.0 $ & $63.0 \pm 16.23 $ & $540.0 \pm 61.62 $ & $120.2 \pm 26.35 $ \\
 & Amplify & $1000.0 \pm 0.0 $ & $63.2 \pm 17.46 $ & $503.4 \pm 70.35 $ & $119.2 \pm 28.15 $ \\
 & Bias & $90.4 \pm 23.71 $ & $159.0 \pm 7.71 $ & $101.4 \pm 14.45 $ & $369.6 \pm 7.4 $ \\
 & Dropout & $622.8 \pm 353.17 $ & $103.8 \pm 44.78 $ & $773.4 \pm 327.12 $ & $221.0 \pm 75.06 $ \\
 & Flip & $39.2 \pm 5.31 $ & $69.8 \pm 16.08 $ & $43.0 \pm 6.63 $ & $49.8 \pm 7.4 $ \\
 & Focus & $60.0 \pm 17.97 $ & $550.2 \pm 293.68 $ & $62.2 \pm 14.75 $ & $93.4 \pm 18.68 $ \\
 & Invert & $31.0 \pm 5.61 $ & $50.0 \pm 8.22 $ & $31.8 \pm 5.17 $ & $34.4 \pm 5.37 $ \\
 & Support & $44.8 \pm 14.01 $ & $107.2 \pm 64.5 $ & $49.2 \pm 19.46 $ & $90.0 \pm 86.16 $ \\
 & Swap & $38.0 \pm 3.46 $ & $66.0 \pm 10.79 $ & $39.6 \pm 3.21 $ & $47.6 \pm 4.56 $ \\
\cline{1-6}
\bottomrule
\end{tabular}

    \caption{First verdict of the differential alignment monitor averaged over 5 runs each 1000 steps.}
    \label{tab:term_1}
\end{table}

\begin{table}[p]
    \small
    \begin{tabular}{llllll}
\toprule
 & Model & Actual & Black  Box & Expert & Gray  Box \\
Benchmark & Predictor &  &  &  &  \\
\midrule
\multirow[t]{9}{*}{Leader (3,5)} & Additive & $84.8 \pm 18.09 $ & $125.0 \pm 1.73 $ & $94.8 \pm 4.15 $ & $351.4 \pm 3.44 $ \\
 & Amplify & $1000.0 \pm 0.0 $ & $54.6 \pm 10.41 $ & $497.4 \pm 70.0 $ & $115.4 \pm 25.52 $ \\
 & Bias & $89.6 \pm 22.98 $ & $126.0 \pm 7.07 $ & $101.4 \pm 14.45 $ & $389.4 \pm 0.55 $ \\
 & Dropout & $1000.0 \pm 0.0 $ & $65.6 \pm 20.44 $ & $719.8 \pm 92.61 $ & $143.4 \pm 35.33 $ \\
 & Flip & $34.2 \pm 4.92 $ & $55.6 \pm 10.41 $ & $35.8 \pm 4.02 $ & $40.0 \pm 4.47 $ \\
 & Focus & $53.6 \pm 10.41 $ & $1000.0 \pm 0.0 $ & $57.0 \pm 5.48 $ & $91.2 \pm 1.64 $ \\
 & Invert & $34.2 \pm 4.92 $ & $55.6 \pm 10.41 $ & $35.8 \pm 4.02 $ & $40.0 \pm 4.47 $ \\
 & Support & $41.8 \pm 6.61 $ & $113.2 \pm 20.66 $ & $44.0 \pm 6.28 $ & $57.6 \pm 8.08 $ \\
 & Swap & $33.2 \pm 2.68 $ & $49.6 \pm 3.58 $ & $34.0 \pm 2.24 $ & $39.0 \pm 2.24 $ \\
\cline{1-6}
\multirow[t]{9}{*}{Nand  (5,2)} & Additive & $84.4 \pm 9.94 $ & $561.6 \pm 10.6 $ & $79.8 \pm 9.36 $ & $130.8 \pm 2.59 $ \\
 & Amplify & $1000.0 \pm 0.0 $ & $84.8 \pm 10.11 $ & $642.0 \pm 27.35 $ & $156.8 \pm 18.54 $ \\
 & Bias & $103.4 \pm 20.56 $ & $156.6 \pm 19.3 $ & $106.8 \pm 14.92 $ & $701.8 \pm 75.35 $ \\
 & Dropout & $168.4 \pm 31.8 $ & $509.8 \pm 275.83 $ & $210.6 \pm 37.56 $ & $485.8 \pm 174.48 $ \\
 & Flip & $55.0 \pm 6.89 $ & $116.8 \pm 10.11 $ & $56.0 \pm 6.2 $ & $73.6 \pm 6.95 $ \\
 & Focus & $78.2 \pm 9.12 $ & $1000.0 \pm 0.0 $ & $74.6 \pm 8.82 $ & $108.2 \pm 3.11 $ \\
 & Invert & $42.6 \pm 5.13 $ & $76.6 \pm 9.45 $ & $42.2 \pm 4.15 $ & $52.0 \pm 3.39 $ \\
 & Support & $77.0 \pm 5.48 $ & $464.2 \pm 132.07 $ & $82.4 \pm 5.37 $ & $133.0 \pm 9.7 $ \\
 & Swap & $40.2 \pm 4.6 $ & $61.8 \pm 9.36 $ & $40.0 \pm 3.94 $ & $45.4 \pm 3.29 $ \\
\cline{1-6}
\multirow[t]{9}{*}{Quantiles} & Additive & $1000.0 \pm 0.0 $ & $82.2 \pm 20.09 $ & $525.6 \pm 64.86 $ & $128.8 \pm 33.57 $ \\
 & Amplify & $1000.0 \pm 0.0 $ & $92.8 \pm 18.38 $ & $594.0 \pm 57.2 $ & $136.2 \pm 32.06 $ \\
 & Bias & $104.2 \pm 28.17 $ & $482.8 \pm 24.35 $ & $114.6 \pm 21.89 $ & $390.2 \pm 11.01 $ \\
 & Dropout & $1000.0 \pm 0.0 $ & $107.0 \pm 33.29 $ & $929.4 \pm 96.12 $ & $184.8 \pm 70.94 $ \\
 & Flip & $46.4 \pm 8.79 $ & $69.8 \pm 22.12 $ & $47.6 \pm 9.76 $ & $54.2 \pm 10.76 $ \\
 & Focus & $94.8 \pm 24.43 $ & $730.0 \pm 62.82 $ & $103.8 \pm 18.29 $ & $307.4 \pm 14.17 $ \\
 & Invert & $40.2 \pm 9.31 $ & $61.0 \pm 14.88 $ & $42.4 \pm 9.79 $ & $47.8 \pm 10.35 $ \\
 & Support & $210.0 \pm 37.44 $ & $156.0 \pm 14.56 $ & $281.2 \pm 41.92 $ & $491.6 \pm 60.16 $ \\
 & Swap & $37.8 \pm 4.76 $ & $48.6 \pm 7.83 $ & $38.8 \pm 4.97 $ & $44.8 \pm 6.02 $ \\
\cline{1-6}
\multirow[t]{9}{*}{crowds-4-3} & Additive & $105.0 \pm 17.25 $ & $289.0 \pm 18.64 $ & $104.6 \pm 17.34 $ & $168.8 \pm 12.54 $ \\
 & Amplify & $1000.0 \pm 0.0 $ & $97.2 \pm 15.42 $ & $757.2 \pm 78.56 $ & $203.8 \pm 22.84 $ \\
 & Bias & $110.8 \pm 23.38 $ & $229.0 \pm 34.83 $ & $108.0 \pm 14.98 $ & $458.8 \pm 18.13 $ \\
 & Dropout & $167.8 \pm 73.16 $ & $319.4 \pm 182.59 $ & $198.2 \pm 108.17 $ & $573.4 \pm 433.51 $ \\
 & Flip & $70.0 \pm 14.07 $ & $127.4 \pm 37.63 $ & $72.0 \pm 15.46 $ & $88.6 \pm 23.16 $ \\
 & Focus & $90.8 \pm 13.77 $ & $1000.0 \pm 0.0 $ & $91.2 \pm 13.48 $ & $110.4 \pm 13.09 $ \\
 & Invert & $46.4 \pm 3.65 $ & $63.6 \pm 14.66 $ & $44.6 \pm 4.72 $ & $48.8 \pm 4.27 $ \\
 & Support & $93.6 \pm 22.96 $ & $425.4 \pm 246.37 $ & $102.6 \pm 28.34 $ & $249.4 \pm 175.94 $ \\
 & Swap & $56.4 \pm 8.5 $ & $95.4 \pm 20.6 $ & $57.8 \pm 10.4 $ & $67.8 \pm 12.36 $ \\
\cline{1-6}
\bottomrule
\end{tabular}

    \caption{First verdict of the differential alignment monitor averaged over 5 runs each 1000 steps.}
    \label{tab:term_2}
\end{table}

\end{document}